\documentclass{scrartcl}

\usepackage{microtype}

\usepackage[sort,numbers]{natbib}
\usepackage{fullpage}
\usepackage[pagebackref]{hyperref}
\usepackage{amsthm}

\usepackage{amsmath}
\usepackage{amssymb}
\usepackage{comment}
\usepackage{amsfonts}

\usepackage{tikz}
\usetikzlibrary{positioning,backgrounds,patterns,calc,shapes}
\tikzstyle{vertex}=[circle,draw=black,minimum size=8pt,inner sep=1pt]
\tikzstyle{vertex2}=[circle,draw=black,minimum size=15pt,inner sep=2pt]
\tikzstyle{edge}=[]
\tikzstyle{ypath}=[ultra thick]
\tikzstyle{dottedEdge}=[dotted,thick]
\tikzstyle{small-vertex}=[circle,draw=black,minimum size=6pt,inner sep=0pt,fill=white]
\tikzstyle{thinedges}=[draw=gray!30]
 \tikzstyle{boxes}=[draw,thick, rounded corners=3mm,text width=2.7cm,align=center,text opacity=1,fill opacity=1,fill=white]
\tikzstyle{unk}=[fill=gray!25!white]

\pgfdeclarelayer{bg}
\pgfsetlayers{bg,main} 

\usepackage{bbding}
\usepackage{placeins}

{\bfseries}{\normalfont}
\theoremstyle{plain}
\newtheorem{theorem}{Theorem}{\bfseries}{\normalfont}
\newtheorem{obs}{Observation}{\bfseries}{\normalfont}
\newtheorem{cor}{Corollary}{\bfseries}{\normalfont}
\newtheorem{prop}{Proposition}{\bfseries}{\normalfont}
\newtheorem{lem}{Lemma}{\bfseries}{\normalfont}

\theoremstyle{definition}
{\bfseries}{\normalfont}

\DeclareMathOperator{\tw}{tw}

\definecolor{ourblue}{RGB}{135,206,250}
\definecolor{ourgreen}{RGB}{0,100,0}
\definecolor{ourred}{RGB}{176,23,31}
\definecolor{lipicsyellow}{rgb}{0.99,0.78,0.07}

\newcommand{\decprob}[3]{%
  \begin{center}%
    \begin{minipage}{0.9\linewidth}%
      \textsc{#1}\\
      \textbf{Input:} #2\\
      \textbf{Question:} #3
    \end{minipage}%
  \end{center}%
}

\usepackage{microtype,ellipsis}

\usepackage{units}

\usepackage[numbers,sort]{natbib}
\makeatletter
\def\NAT@spacechar{~}
\makeatother

\usepackage{paralist}

\usepackage{enumitem}
\usepackage[capitalize]{cleveref}

\newlist{propertylist}{enumerate}{10}
\setlist[propertylist]{label=\roman*),leftmargin=*}
\crefname{propertylisti}{Property}{Properties}

\newcommand{\MCC}{\textsc{Multicolored Clique}\xspace}

\newcommand{\CVC}{\textsc{Cubic Vertex Cover}\xspace}
\newcommand{\VC}{\textsc{Vertex Cover}\xspace}
\newcommand{\FVS}{\textsc{Feedback Vertex Set}\xspace}
\newcommand{\CKC}{\textsc{Collapsed $k$-Core}\xspace}
\newcommand{\DVD}{\textsc{$r$-Degenerate Vertex Deletion}\xspace}

\newcommand{\coNP}{\ensuremath{\textsf{coNP}}\xspace}
\newcommand{\NP}{\ensuremath{\textsf{NP}}\xspace}
\newcommand{\pNP}{\ensuremath{\textsf{para-NP}}\xspace}

\newcommand{\FPT}{\ensuremath{\textsf{FPT}}\xspace}
\newcommand{\W}[1]{\ensuremath{\textsf{W[#1]}}\xspace}

\crefname{rrule}{Rule}{Rules}

\DeclareRobustCommand{\NoKernelAssume}{$\NP\subseteq \text{\coNP/poly}$}

\graphicspath{{images/}}

\crefname{lem}{Lemma}{Lemmata}
\Crefname{lem}{Lem}{Lems}
\crefname{proposition}{Proposition}{Proposition}
\Crefname{proposition}{Prop}{Props}
\crefname{theorem}{Theorem}{Theorems}
\Crefname{theorem}{Thm}{Thms}
\crefname{cor}{Corollary}{Corollaries}
\Crefname{cor}{Cor}{Cors}
\crefname{observation}{Observation}{Obervations}
\Crefname{observation}{Obs}{Obs}
\crefname{definition}{Definition}{Definitions}
\Crefname{definition}{Def}{Defs}
\crefname{figure}{Figure}{Figures}
\Crefname{figure}{Fig}{Figs}
\crefname{section}{Section}{Sections}
\Crefname{section}{Sec}{Secs}

\usepackage{etoolbox}
\usepackage[linesnumbered, ruled]{algorithm2e}

\usepackage{authblk}

\title{A Parameterized Complexity View on Collapsing~$k$-Cores}

\author[1,2,3]{Junjie Luo\thanks{Supported by
CAS-DAAD Joint Fellowship Program for Doctoral Students of UCAS.}}

\affil[1]{\small Institut f\"ur Softwaretechnik und Theoretische Informatik, TU Berlin, {Berlin, Germany}, 
 \texttt{junjie.luo@campus.tu-berlin.de, h.molter@tu-berlin.de}}
\affil[2]{\small Academy of Mathematics and Systems Science, Chinese Academy of Sciences, Beijing, China}
\affil[3]{\small School of Mathematical Sciences, University of Chinese Academy of Sciences, Beijing, China}

\author[1]{Hendrik~Molter\thanks{Supported by the DFG, project MATE (NI 369/17).}}

\author[4]{Ond\v rej~Such\'{y}\thanks{Supported by grant 17-20065S of the Czech Science Foundation.}}
\affil[4]{\small Department of Theoretical Computer Science, Faculty of Information Technology, Czech Technical University in Prague, Prague, Czech~Republic, \texttt{ondrej.suchy@fit.cvut.cz}}

\date{}

\begin{document}

\maketitle

\begin{abstract}
We study the \NP-hard graph problem \CKC\ where, given an undirected graph~$G$ and integers~$b$,~$x$, and~$k$, we are asked to remove~$b$ vertices such that the $k$-core of remaining graph, that is, the (uniquely determined) largest induced subgraph with minimum degree $k$, has size at most~$x$. \CKC\ was introduced by Zhang et al.~[AAAI~2017] and it is motivated by the study of engagement behavior of users in a social network and measuring the resilience of a network against user drop outs. \CKC\ is a generalization of \DVD\ (which is known to be \NP-hard for all $r\ge0$) where, given an undirected graph~$G$ and integers~$b$ and~$r$, we are asked to remove~$b$ vertices such that the remaining graph is $r$-degenerate, that is, every its subgraph has minimum degree at most $r$. 

We investigate the parameterized complexity of \CKC\ with respect to the parameters $b$, $x$, and $k$, and several structural parameters of the input graph. We reveal a dichotomy in the computational complexity of \CKC\ for $k\le 2$ and $k\ge 3$. For the latter case it is known that for all $x\ge 0$ \CKC\ is \W{P}-hard when parameterized by $b$. We show that \CKC\ is \W{1}-hard when parameterized by $b$ and in \FPT\ when parameterized by $(b+x)$ if $k\le 2$. Furthermore, we show that \CKC\ is in \FPT\ when parameterized by the treewidth of the input graph and presumably does not admit a polynomial kernel when parameterized by the vertex cover number of the input graph.

\bigskip

\end{abstract}

\section{Introduction}
In recent years, modelling user engagement in social networks has received substantial interest~\cite{zhang2017finding,bhawalkar2015preventing,chitnis2013preventing,wang2016recommending,garcia2013social,wu2013arrival}. A popular assumption is that a user engages in a social network platform if she has at least a certain number of contacts, say $k$, on the platform. Further, she is inclined to abandon the social network if she has less than $k$ contacts~\cite{zhang2017finding,bhawalkar2015preventing,chitnis2013preventing,malliaros2013stay,garcia2013social,chitnis2016parameterized}. In compliance with this assumption, a suitable graph-theoretic model for the ``stable'' part of a social network is the so-called \emph{$k$-core} of the social network graph, that is, the largest induced subgraph with minimum degree $k$~\cite{Sei83}.\footnote{Note that the $k$-core of a graph is uniquely determined.} 

Now, given a stable social network, that is, a graph with minimum degree $k$, the departure of a user decreases the degree of her neighbors in the graph by one which then might be smaller than $k$ for some of them. Following our assumption these users now will abandon the network, too. This causes a cascading effect of users dropping out (collapse) of the network until a new stable state is reached.

From an adversarial perspective a natural question is how to maximally destabilize a competing social network platform by compelling $b$ users to abandon the network. This problem was introduced as \CKC\ by Zhang et al.~\cite{zhang2017finding} and the decision version is formally defined as follows.
\decprob{\CKC}{An undirected graph $G=(V,E)$, and integers $b$, $x$, and
$k$.}{Is there a set $S\subseteq V$ with $|S|\le b$ such that the $k$-core of
$G-S$ has size at most $x$?}
In the mentioned motivation, one would aim to minimize $x$ for a given $b$ and $k$. Alternatively, we can also interpret this problem as a measure for resilience agains user drop outs of a social network by determining the smallest $b$ for a given $k$ and $x$.

\subparagraph{Related Work.} In 2017 Zhang et al.~\cite{zhang2017finding} showed that \CKC\ is \NP-hard for any $k\ge 1$ and gave a greedy algorithm to compute suboptimal solutions for the problem. However, for $x=0$ and any fixed $k$, solving \CKC\ is equivalent to finding $b$ vertices such that after removing said vertices, the remaining graph is $(k-1)$-degenerate\footnote{A graph $G$ is $r$-degenerate if every subgraph of $G$ has a vertex with degree at most $r$~\cite{diestel2016graphentheory}.}.
This problem is known as \DVD\ and it is defined as follows.
\decprob{\DVD}
{An undirected graph $G=(V,E)$, and integers $b$ and $r$.}%
{Is there a set $S\subseteq V$ with $|S|\le b$ such that $G-S$ is $r$-degenerate?}
It is easy to see that \CKC\ is a generalization of \DVD. In 2010 Mathieson~\cite{mathieson2010parameterized} showed that \DVD\ is \NP-complete and \W{P}-complete when parameterized by the budget $b$ for all~$r\ge2$ even if the input graph is already $(r+1)$-degenerate and has maximum degree $2r+1$. In the mid-90s Abrahamson et al.~\cite{AbrahamsonDF95} already claimed \W{P}-completeness for \DVD\ with $r=2$ when parameterized by $b$ under the name \textsc{Degree 3 Subgraph Annihilator}. For $r=1$, the problem  is equivalent to \textsc{Feedback Vertex Set} and for~$r=0$  it is equivalent to \textsc{Vertex Cover}, both of which are known to be \NP-complete and fixed-parameter tractable when parameterized by the solution size~\cite{GJ79,cygan2015parameterized}. The aforementioned results concerning \DVD\ in fact imply the hardness results shown by Zhang et al.~\cite{zhang2017finding} for \CKC.

\subparagraph{Our Contribution.} We complete the parameterized complexity landscape of \CKC\ with respect to the parameters $b$, $k$, and $x$. Specifically, we correct errors in the literature~\cite{mathieson2010parameterized,AbrahamsonDF95} concerning the \W{P}-completeness of \DVD\ when parameterized by $b$ for $r=2$. We go on to clarify the parameterized complexity of \CKC\ for $k\le 2$ by showing \W{1}-hardness for parameter $b$ and fixed-parameter tractability for the combination of $b$ and $x$. Together with previously known results, this reveals a dichotomy in the computational complexity of \CKC\ for $k\le 2$ and~$k\ge 3$. 

We present two single exponential linear time \FPT\ algorithms, one for \CKC\ with $k=1$ and one for $k=2$. In both cases the parameter is $(b+x)$. In particular, the algorithm for $k=2$ runs in $O(1.755^{x+4b}\cdot n)$ time which means that it solves \textsc{Feedback Vertex Set} in $O(9.487^b\cdot n)$ time (here, $b$ is the solution size of \textsc{Feedback Vertex Set}). To the best of our knowledge, despite of its simplicity our algorithm improves many previous linear time parameterized algorithm for \textsc{Feedback Vertex Set}~\cite{lokshtanov2018linear,guo2006compression}. However, recently a linear time computable polynomial kernel for \textsc{Feedback Vertex Set} was shown~\cite{Iwata17}, which together with e.g.~\cite{kociumaka2014faster} yields an even faster linear time parameterized algorithm.

Furthermore, we conduct a thorough parameterized complexity analysis with respect to structural parameters of the input graph. On the positive side, we show that \CKC\ is fixed-parameter tractable when parameterized by the treewidth of the input graph and show that it presumably does not admit a polynomial kernel when parameterized by either the vertex cover number or the bandwidth of the input graph. We also show that the problem is fixed-parameter tractable when parameterized by the combination of the cliquewidth of the input graph and $b$ or $x$.
Further results include \W{1}-hardness when parameterized by the clique cover number of the input graph and \pNP-hardness for the domination number of the input graph.

\section{Hardness Results from the Literature}
\label{sec:hardness}
In this section, we gather and discuss known hardness results for \CKC. 
Recall that that \CKC\ with~$x=0$ is the same problem as \DVD\ with $r=k-1$.
Hence, the hardness of \CKC\ was first established by Mathieson~\cite{mathieson2010parameterized} who showed that \DVD\ is \NP-complete and \W{P}-complete when parameterized by the budget $b$ for all~$r\ge2$ even if the input graph is already $(r+1)$-degenerate and has maximum degree $2r+1$. However, in the proof of Mathieson~\cite{mathieson2010parameterized} the reduction is incorrect for the case that $r=2$. Abrahamson et al.~\cite{AbrahamsonDF95} claim \W{P}-completeness for $r=2$ but their reduction is also flawed. In the following, we provide counterexamples for both cases. Then we show in the following how to adjust the reduction of Mathieson~\cite{mathieson2010parameterized} for this case.

\subparagraph{Counterexample for Mathieson's Reduction~\cite{mathieson2010parameterized}.} We refer to the original paper by Mathieson~\cite{mathieson2010parameterized} for definitions, notation, and description of the gadgets and the reduction itself. Mathieson provides a reduction from \textsc{Cyclic Monotone Circuit Activation} to \DVD~\cite[Theorem~4.4]{mathieson2010parameterized} showing that \DVD\ is \W{P}-complete when parameterized by the budget~$b$ for all~$r\ge2$ even if the input graph is already $(r+1)$-degenerate and has maximum degree~$2r+1$. However, for the case of $r=2$ it is easy to see that the OR gadget is already 2-degenerate. We illustrate the flawed OR gadget for $r=2$ in \autoref{fig:mathiesonOR}.
\begin{figure}[t]
\begin{center}
\begin{tikzpicture}[every node/.style={draw, fill=gray, circle}, line width=.6pt]
 \node[label={[label distance=0pt]west:$v^i$}] (1) at (0,4.5) {};
 \node (2) at (0,3.5) {};
 \node[fill=ourred] (3) at (0,2) {};
 \node[label={[label distance=0pt]west:$v^o$}] (4) at (0,1) {};
 \draw[dashed] (-1,5.5) -- (1);
 \draw[dashed] (1,5.5) -- (1);
 \draw[dashed] (-1,0) -- (4);
 \draw[dashed] (1,0) -- (4);
 \draw (1) -- (2);
 \draw (3) -- (4);
 \draw[color=gray] (2) -- (3);
 \path (2) edge[color=gray, bend right] (4);
\end{tikzpicture}
\end{center}
\caption{Illustration of the OR gadget in the reduction of Mathieson~\cite[Theorem~4.4]{mathieson2010parameterized} for $r=2$. Note that the red colored vertex has degree 2 and the whole gadget is 2-degenerate.}\label{fig:mathiesonOR}
\end{figure}
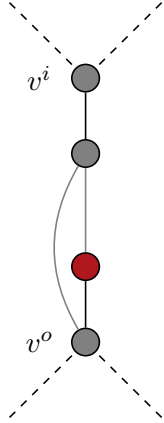

This means that whenever a \textsc{Cyclic Monotone Circuit Activation} instance is activated by the set of all its binary OR gates, the graph produced by the reduction (for $r=2$) is already 2-degenerate, which clearly makes the reduction incorrect. An example of such an instance would be the one given by Mathieson~\cite{mathieson2010parameterized}. We reproduce the example in \autoref{fig:mathiesonExample}.

\begin{figure}[t]
\begin{center}
\begin{tikzpicture}[every node/.style={draw, circle,minimum width=1.1cm},line width=.6pt]
 \node[label={[label distance=-10pt]west:$c$}] (1) at (-2,-2) {AND};
 \node[label={[label distance=-10pt]west:$a$}] (2) at (0,-.8) {OR};
 \node[label={[label distance=-10pt]east:$d$}] (3) at (2,-2) {OR};
 \node[label={[label distance=-10pt]north:$b$}] (4) at (0,1) {AND};
 \draw[->] (1) -- (4);
 \draw[->] (3) -- (1);
 \draw[->] (3) -- (4);
 \draw[->] (4) -- (2);
 \draw[->] (2) -- (1);
 \draw[->] (2) -- (3);
\end{tikzpicture}
\end{center}
\caption{An example instance of \textsc{Cyclic Monotone Circuit Activation} given by Mathieson~\cite{mathieson2010parameterized}. The sets $\{a\}$ and $\{b\}$ activate the entire circuit, whereas the sets $\{c\}$ and $\{d\}$ do not. In particular, the set $\{a,d\}$ of all OR~gates activates the entire circuit.}\label{fig:mathiesonExample}
\end{figure}
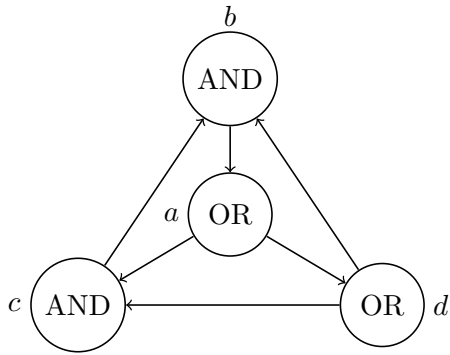

We describe how to repair the reduction (for $r=2$) in the proof of \autoref{thm:mathieson}.
\subparagraph{Counterexample for Abrahamson et al.'s Reduction~\cite{AbrahamsonDF95}.} We refer to the original paper by Abrahamson et al.~\cite{AbrahamsonDF95} for definitions, notation, and description of the gadgets and the reduction itself. The same reduction (using the same notation) can also be found in the book ``Fundamentals of Parameterized Complexity'' by Downey and Fellows~\cite{DF13}. Abrahamson et al.\ provide a reduction from \textsc{Weighted Monotone Circuit Satisfiability} to \textsc{Degree~3 Subgraph Annihilator}, which is equivalent to \DVD\ with~$r=2$. With this reduction they claim to show that \DVD\ with $r=2$ is \W{P}-complete when parameterized by the budget~$b$~\cite[Theorem~3.7~(ii)]{AbrahamsonDF95}. 

The main idea of their reduction is that once a satisfying assignment is found, all variable gadgets corresponding to variables that are set to false are also removed. However, since in a variable gadget for a variable $x[i]$, the vertex $v(i,4)$ has high degree, it is only removed if sufficiently many of the gate gadgets that it is connected to are also removed. However, an AND gadget combined with a fan-out gadget is only removed if both inputs are removed. This follows from fan-out gadget not being removable from below. This allows us to create a counterexample with $k=1$, which we illustrate in \autoref{fig:abrahamsonExample}. It is easy to check that~$x[1]=x[2]=\text{false}, x[3]=\text{true}$ is the only satisfying assignment that has at most one variable set to true. Furthermore, the AND gates in the red area have two outgoing connections each, hence the corresponding AND gadgets have fan-out gadgets attached to them. Initially, the output of these gates is false for the satisfying assignment so the fan-out gadget attached to them are only removed if the AND gadgets themselves are removed. An AND gadget is removed if both of its inputs are removed. However, note that the variable gadget for~$x[1]$ is not completely removed after~$v(3,4)$ is removed from the graph. In particular, the vertex~$v(1,4)$ has still degree $m(k+1)>3$ after all vertices are removed since the AND gadgets it connects to are not removed. It follows that the graph is not 2-degenerate after removing~$v(3,4)$ and hence the reduction is not correct.

We believe that the reduction can be corrected by replacing the high degree vertices~$v(i,4)$ by fan-out gadgets that connect to the gate gadgets. However, we omit a proof of this claim.
\begin{figure}[t]
\begin{center}
\tikzstyle{gate}=[draw, circle,minimum width=.6cm,fill=white]
\begin{tikzpicture}[line width=.6pt]
\node[rectangle, fill=ourred!40, draw=ourred, draw, minimum width=3cm, minimum height=1cm, rounded corners=2mm] (gg) at (0, -1.5) {};
\node (A) at (-2,0) {$x[1]$};
\node (B) at (0,0) {$x[2]$};
\node (C) at (2,0) {$x[3]$};
\node[gate] (g1) at (-1,-1.5) {$\wedge$};
\node[gate] (g2) at (1,-1.5) {$\wedge$};
\node[gate] (g3) at (-1,-3) {$\wedge$};
\node[gate] (g4) at (1,-3) {$\vee$};
\node[gate] (g5) at (1.5,-4) {$\vee$};
\node[gate] (g6) at (0,-5) {$\vee$};
\node (O) at (0,-6) {output};
\draw[->] (A) -- (g1);
\draw[->] (A) -- (g2);
\draw[->] (B) -- (g1);
\draw[->] (C) -- (g2);
\draw[->] (g1) -- (g3);
\draw[->] (g1) -- (g4);
\draw[->] (g2) -- (g3);
\draw[->] (g2) -- (g4);
\draw[->] (C) -- (g5);
\draw[->] (g4) -- (g5);
\draw[->] (g3) -- (g6);
\draw[->] (g5) -- (g6);
\draw[->] (g6) -- (O);
\node[rectangle, fill=none, draw=ourred, draw, minimum width=3cm, minimum height=1cm, rounded corners=2mm] (gg) at (0, -1.5) {};
\end{tikzpicture}
\end{center}
\caption{A \textsc{Weighted Monotone Ciruit Satisfiability} instance with $k=1$ (max.\ number of variables set to true). It is easy to check that $x[1]=x[2]=\text{false}, x[3]=\text{true}$ is the only satisfying assignment that has at most one variable set to true. The AND gates in the red area have two outgoing connections each, hence the corresponding AND gadgets have fan-out gadgets attached to them.}\label{fig:abrahamsonExample}
\end{figure}
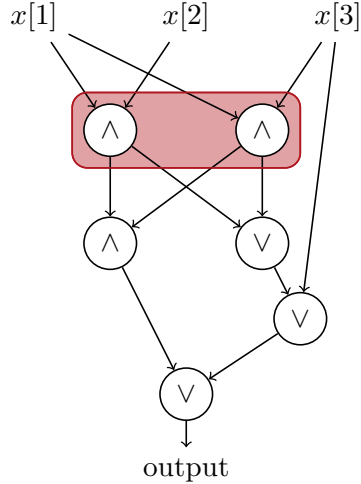

\begin{theorem}[Corrected from \cite{mathieson2010parameterized}]\label{thm:mathieson}
For any $r \ge 2$ \DVD\ is \NP-hard and \W{P}-complete when parameterized by $b$, even if the
degeneracy of the input graph is $r+1$ and the maximum degree of the input graph
is $2r+1$.
\end{theorem}
\begin{proof}
We refer to the original paper by Mathieson~\cite{mathieson2010parameterized} for definitions, notation, and description of the gadgets and the reduction itself. Mathieson provides a reduction from \textsc{Cyclic Monotone Circuit Activation} to \DVD~\cite[Theorem~4.4]{mathieson2010parameterized} showing that \DVD\ is \W{P}-complete when parameterized by the budget~$b$ for all~$r\ge2$ even if the input graph is already $(r+1)$-degenerate and has maximum degree~$2r+1$. While Mathieson~\cite{mathieson2010parameterized} claimed the result also for $r=2$, the reduction is incorrect in this case. Similarly, while Abrahamson et al.~\cite{AbrahamsonDF95} claimed the same result for $r=2$ under the name \textsc{Degree 3 Subgraph Annihilator}, their reduction also seems to be flawed. 

However, there is a way to correct the proof of Mathieson~\cite{mathieson2010parameterized}:
The only problem is that his OR gadget is 2-degenerate, so the graph sometimes collapses without even deleting a single vertex. Before we introduce the correct gadget, note that the gadget is always used with exactly two inputs (predecessor gates).
We can replace the OR gadget for case $r=2$ with the graph illustrated in \autoref{fig:mathieson}.
\begin{figure}[t]
\begin{center}
\begin{tikzpicture}[every node/.style={draw, fill=gray, circle}, line width=.6pt]
 \node (1) at (0,4) {};
 \node (2) at (0,3) {};
 \node (3) at (0,2) {};
 \node (4) at (0,1) {};
 \node (5) at (1,4) {};
 \node (6) at (1,3) {};
 \node (7) at (1,2) {};
 \node[label={[label distance=0pt]west:$v^o$}] (8) at (1,0) {};
 \draw[dashed] (-1,5) to node[left,draw=none, fill= none,rectangle, inner sep=10pt] {\small input from predecessor 1} (1);
 \draw[dashed] (2,5) to node[right,draw=none, fill= none,rectangle, inner sep=10pt] {\small input from predecessor 2} (5);
 \draw (1) to (2);
 \draw (1) to (6);
 \draw (5) to (2);
 \draw (5) to (6);
 \draw (2) to (3);
 \draw (6) to (7);
 \draw (3) to (4);
 \draw (3) to (8);
 \draw (7) to (4);
 \draw (7) to (8);
 \draw (4) to (8);
 \draw[dashed] (8) to (0,-1);
 \draw[dashed] (8) to (2,-1);
\node[rectangle, fill=none, draw=ourred, dotted, minimum width=1.6cm, minimum height=.6cm, rounded corners=2mm,label={[label distance=0pt]west:$v^i$}] (vi) at (.5, 4) {};
\end{tikzpicture}
\end{center}
\caption{Illustration of the corrected OR gadget for $r=2$ in the proof of \autoref{thm:mathieson}. The two vertices surrounded by the red dotted line play the role of $v^i$ in the original version of the gadget.}\label{fig:mathieson}
\end{figure}
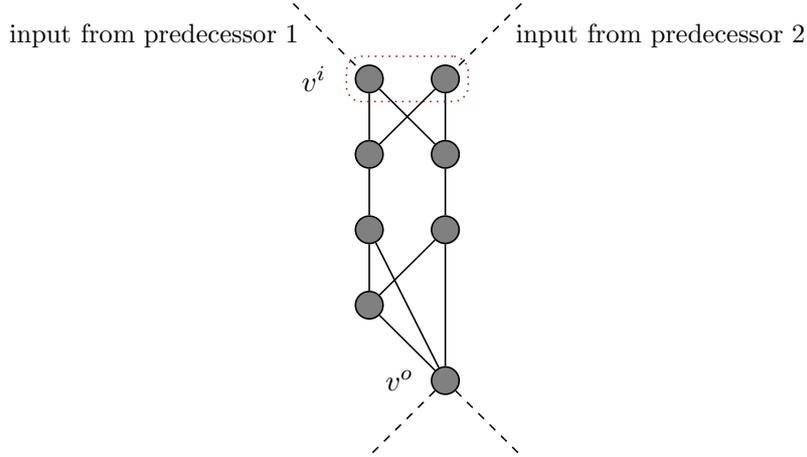
To collapse this gadget one can simply delete $v^{o}$. The correctness now follows from an analogous argument as given by Mathieson~\cite{mathieson2010parameterized}.
\end{proof}
The following observation shows that the hardness result by Mathieson~\cite{mathieson2010parameterized} (\autoref{thm:mathieson}) easily transfers to \CKC\ (also in the cases where $x\neq 0$).
\begin{obs}\label{obs:non_empty_core}
 Let $x'>0$ be a positive integer.
 There is a reduction which transforms instances $(G,b,x,k)$ of \CKC\ with $x=0$ into equivalent instances $(G',b,x',k)$ of \CKC{}.
\end{obs}
\begin{proof}
 We distinguish two cases, depending on the relation of $x'$ to $k$.
 
 If $x'\le k$, then we let $G'=G$. Obviously, if $S$ is a solution for $(G,b,x,k)$, then it is also a solution for $(G',b,x',k)$. On the other hand, if $S'$ is a solution for $(G',b,x',k)$, then $G' \setminus S' = G \setminus S'$ has a $k$-core with at most $k$ vertices. However, any vertex of a graph with at most $k$ has degree at most $k-1$ and, thus, the $k$-core is empty. Therefore $S'$ is a solution for $(G,b,x,k)$.
 
 If $x' \ge k+1$, then we obtain $G'$ as a disjoint union of $G$ and a clique $C$ on $x'$ vertices.
 Again obviously, if $S$ is a solution for $(G,b,x,k)$, then it is also a solution for $(G',b,x',k)$, since the $k$-core of $G'\setminus S$ is exactly $C$. 
 
 Now let $S'$ be a solution for $(G',b,x',k)$. Note that if one vertex of $C \setminus S'$ is part of the $k$-core of $G \setminus S'$, then all vertices of $C \setminus S'$ are. If indeed $C \setminus S'$ is a part of the $k$-core of $G \setminus S'$, then the $k$-core contains at most $C \cap S'$ other vertices. If $Q$ is the set of vertices in the $k$-core of $G \setminus (S' \cup C)= G \setminus S'$, then $Q \cup (S' \setminus C)$ is a solution for $(G,b,x,k)$.
 
 Now if $C \setminus S'$ is not a part of the $k$-core of $G \setminus S'$, then we know that $|C \cap S'|$ vertices are sufficient to collapse a clique of size $x'$. Since the $k$-core of $G' \setminus S'$, which is the same as the $k$-core of $G \setminus S'$ is of size at most $x'$, there is a set $Q$ of vertices at most $|C \cap S'|$ such that the $k$-core of $G \setminus (S' \cup Q)$ is empty. Hence $(S' \cup Q)$ is a solution for $(G,b,x,k)$. 
\end{proof}
With that, we arrive at the following corollary.
\begin{cor}\label{cor:whardness}
\CKC\ is \NP-hard and \W{P}-hard when parameterized by $b$ for all $x\ge 0$ and $k\ge 3$, even if the
degeneracy of the input graph is $\max\{k,x-1\}$ and the maximum degree of the input graph
is $\max\{2k-1,x-1\}$.
\end{cor}
Note that \DVD\ is known to be \NP-hard for all $r\ge 0$.\footnote{\autoref{thm:mathieson} states \NP-hardness of \DVD\ for $r\ge 2$. Recall that for $r=1$ \DVD\ is equivalent to \textsc{Feedback Vertex Set} and for $r=0$ it is equivalent to \textsc{Vertex Cover}, both of which are known to be \NP-hard~\cite{GJ79}.} Hence, we also know that \CKC\ is \NP-hard for $k\le 2$ and all $x\ge 0$. However, the parameterized complexity with respect to $b$ is open in this case. We settle this in the next section.

\section{Algorithms for $k=1$ and $k=2$}\label{sec:smallk}
In this section we investigate the parameterized complexity of \CKC\ for the case that $k\le 2$. Since \cref{cor:whardness} only applies for $k\ge 3$ we first show in the following that the problem is \W{1}-hard with respect to the combination of $b$ and $(n-x)$ for all $k\ge 1$. 
Furthermore, we present two algorithms; one that solves \CKC\ with $k=1$ and one for the $k=2$ case. Both algorithm run in single exponential linear \FPT-time with respect to the parameter combination $(b+x)$. 

We first give a parameterized reduction from \textsc{Clique} to \CKC. Note that since this hardness result holds for the combination of $b$ and the dual parameter of $x$, it is incomparable to \cref{cor:whardness} for~$k\ge 3$.
\begin{prop}\label{thm:whardsmallk}
\CKC\ is \W{1}-hard when parameterized by the combination of $b$ and $(n-x)$ for all $k\ge1$, even if the
input graph is bipartite and $\max\{2,k\}$-degenerate.
\end{prop}

\begin{proof}
We reduce from \W{1}-hard problem \textsc{Clique}~\cite{cygan2015parameterized}, where given a graph $G=(V,E)$ and an integer $p$, the task is to decide whether $G$ contains a clique of size at least $p$. Let $(G,p)$ be an instance of \textsc{Clique} and $k$ be a given constant. We build an instance $(G',b,x,k)$ of \CKC\ as follows. We can assume that $p \le |V(G)|$, as otherwise we can output a trivial no-instance. We further assume that each vertex of $G$ has degree at least $p+1$. Vertex of degree less than $p-1$ is not part of a clique of size at least $p$, while for all vertices of degree $p-1$ or $p$ we can check in $O(n \cdot p^3)$ time whether there is a clique of size at least $p$ containing any of them.

We let $V(G')= V \cup U \cup W$, where $V=V(G)$ are the vertices of $G$, $U=\{u_e^i \mid e \in E,  i \in \{1, \ldots, k\}\}$, and $W=\{w_e^i \mid e \in E,  i \in \{1, \ldots, k-1\}\}$.
We also let $E(G')=E_U \cup E_W$, where $E_U=\{\{v, u^i_e\}\mid e \in E,  i \in \{1, \ldots, k\}, v \in e\}$, and $E_W=\{\{u^i_e,w^{i'}_e\}\mid e \in E,  i \in \{1, \ldots, k\}, i' \in \{1, \ldots, k-1\}\}$. We actually only introduce the sets $W$ and $E_W$ if $k\ge2$.

Finally we set $b=p$ and $x=n-(p+(2k-1)\binom{p}{2})$, where $n=|V(G')|$ is the number of vertices of graph $G'$.

We claim that $(G',b,x,k)$ is a yes-instance of \CKC\ if and only if $(G,p)$ is a yes-instance of \textsc{Clique}. 

$\Rightarrow:$ If $S$ is a clique of size at least $p$ in $G$, then we claim that deleting $S$ from $G'$ results in a $k$-core of size at most $x$. Let $e$ be an edge between two vertices of $S$ in $G$. Then for any $i \in \{1, \ldots, k\}$ vertex $u_e^i$ has only vertices $w_e^1,\ldots w_e^{k-1}$ as neighbors in $G'\setminus S$. Hence, it has degree $k-1$ in this graph and it is not part of the $k$-core of the graph. Since this holds for each $i$, vertices $w_e^1,\ldots w_e^{k-1}$ do not have any neighbors in the $k$-core of $G'\setminus S$ and, hence, they are also not in the $k$-core of the graph. This makes $2k-1$ vertices per each edge of the clique which are not deleted and not in the $k$-core, showing that the size of the $k$-core is at most $x$.

$\Leftarrow:$ Now let $S$ be a set of vertices of $V(G')$ of size at most $b$ such that $G'\setminus S$ has $k$-core of size at most $x$. 
For $e \in E$ let $UW_e=\{u_e^i \mid  i \in \{1, \ldots, k\}\} \cup \{w_e^i \mid i \in \{1, \ldots, k-1\}\}$.
Let $S_E = \{e \in E \mid S \cap UW_e \neq \emptyset\}$. Note that if $e =\{x,y\}$, $e \notin S_E$, and $|S \cap \{x,y\}|\le 1$, then whole $UW_e \cup \{x,y\} \setminus S$ is in the $k$-core of $G'\setminus S$ as each vertex in $UW_e \cup \{x,y\} \setminus S$ has at least $k$ neighbors in $UW_e \cup \{x,y\} \setminus S$. This means that 
each vertex in $V \setminus S$ is in the $k$-core of $G'\setminus S$. Indeed, for an arbitrary vertex $v$ in $V \setminus S$ the degree of $v$ in $G$ is at least $p+1$ and, thus,  there is at least one edge $e$ incident to $v$ which is not in $S_E$. Hence $v$ is in the $k$-core of $G'\setminus S$ by the above argument.

For $e \in S_E$ possibly no vertex of $UW_e$ is in the $k$-core of $G' \setminus S$, effectively shrinking it by $2k-1$ vertices. However, this does not influence the other vertices in $V$, $U$, or $W$, since $V \setminus S$ is in the $k$-core, as we already observed. If $e= \{x,y\}$ and $\{x,y\} \subseteq S$, then also the whole $UW_e$ is not in the $k$-core as observed in the first implication. Thus, if $|S \cap V|=a$ and $|S_E|=c$, then there are at most $(2k-1)(\binom{a}{2}+c)$ vertices of $G'$ which are neither in $S$ nor in the $k$-core of $G' \setminus S$. As $S$ is a solution, this number has to be at least $(2k-1)\binom{p}{2}$, while $a+c \le b=p$. It follows that $a=p$ and $S$ is a clique of size $p$ in $G$.

Note that graph $G'$ is bipartite and for $k \ge 2$ it is also $k$-degenerate, as all vertices in $W$ have degree $k$, after removing them the vertices of $U$ have degree $2$, and, finally, $V$ forms an independent set in $G'$.
\end{proof}

Now we proceed with the algorithm for \CKC\ with $k=1$. While there is a simple algorithm with $O(3^{x+b}(m+n))$ running time\footnote{An informal description of the algorithm: We use an initially empty set $X$ that should contain vertices of the remaining 1-core. We branch over edges where both endpoints are not in $X$ and either remove one of the endpoints or put both endpoint into $X$. Then we branch over all edges that have exactly one endpoint in $X$ and either remove the other endpoint or put the other endpoint into $X$ as well.} for this case, we prefer to present an algorithm with the slightly worse running time as stated, since we then generalize this algorithm to the case $k=2$ with some modifications.

\begin{prop}
\label{thm:fpt_b+x_k=1}
\CKC\ with $k=1$ can be solved in $O(2^{x+2b}(m+n))$ time.
Assuming the Exponential Time Hypothesis, there is no $2^{o(b)+f(x)}n^{O(1)}$ time algorithm for \CKC\ with $k=1$, for any function $f$.
\end{prop}

\begin{algorithm}[t]
\caption{Algorithm for \CKC\ with $k=1$}\label{algo:1}
\textsc{SolveRec}$(G,S,Q)$\\
\lIf{$|S| > b$ or $|Q| > b+x$}{\textbf{return} No solution\label{algst:too_big}}
Let $G'$ be the $1$-core of $G \setminus S$\label{algst:1-core}.\\
\lIf{$|V(G')| \le x$}{\textbf{return} $S$\label{algst:happy}}
\lIf{$V(G') \subseteq Q$}{\textbf{return} No solution\label{algst:stuck}}
Let $v$ be the vertex with the highest degree in $G'$ which is not in $Q$\label{algst:select}\\
$T \leftarrow$ \textsc{SolveRec}$(G,S \cup \{v\},Q)$ \label{algst:rec_start}\\
\lIf{$T\neq \text{\normalfont{No solution}}$}{\textbf{return} $T$}
\lElse{\textbf{return} \textsc{SolveRec}$(G,S,Q \cup \{v\})$\label{algst:rec_end}}  
\end{algorithm}

\emph{Algorithm:} We present a recursive algorithm (see Algorithm \ref{algo:1} for pseudocode) that maintains two sets $S$ and $Q$. 
The recursive function is supposed to return a solution to the instance, whenever there is a solution $B$ containing all of $S$ and there is no solution containing $S$ and anything of $Q$. If some of the conditions is not met, then the function should return ``No solution''.
In other words, $S$ is the set of deleted vertices and $Q$ is the set of vertices the algorithm has decided not to delete in the previous steps but may be collapsed in the future.
Hence, the solution to the instance, or the information that there is none, is obtained by calling the recursive function with both sets $S$ and $Q$ empty.

The algorithm first checks that $S$ is of size at most $b$ and $Q$ is of size at most $x+b$. 
If any of these is not true, then it rejects the current branch.
Then it computes the $1$-core $G'$ of the graph $G \setminus S$.
If the $1$-core is of size at most $x$, then it returns $S$ as a solution.
If $G'$ is larger, but all its vertices are in $Q$, then we have no way to shrink the core and we again reject.
Finally, the algorithm picks an arbitrary vertex $v$ of largest degree in $G'$ which is not in $Q$ and recurses on both possibilities---either $v$ is in the solution, modeled by adding it to $S$, or it is in no solution containing $S$, modeled by adding it to $Q$. 
We start by showing that Algorithm~\ref{algo:1} has the claimed running time.
\begin{lem}\label{lem:alg1runningtime}
Algorithm~\ref{algo:1} runs in $O(2^{x+2b}(m+n))$ time.
\end{lem} 
\begin{proof}
Assume the running time of Algorithm \ref{algo:1} is $T(\mu)$, where $\mu	=|S|+|Q|$.
Let $m$ and $n$ be the number of edges and number of vertices in $G$ respectively.
Since when $|S| > b$ or $|Q| > b+x$ Algorithm \ref{algo:1} directly return ``No solution'' in line \ref{algst:too_big}, we have $T(\mu)=O(1)$ for $\mu>2b+x$.
Lines \ref{algst:too_big}, \ref{algst:happy}, \ref{algst:stuck} and \ref{algst:select} can be done in $O(n)$ time.
In line \ref{algst:1-core} the $1$-core $G'$ can be found in $O(n+m)$ time by first removing vertices in $S$ and edges incident with them to get $G \setminus S$, and then removing isolated vertices in $G \setminus S$. 
Thus except for line \ref{algst:rec_start} and \ref{algst:rec_end}, all steps can be done in $O(m+n)$ time, and we have that
\[
T(0) \le 2T(1)+O(m+n) \le \dots \le 2^{\mu+1}T(\mu+1)+(1 + 2 + \dots + 2^{\mu})O(m+n)=O(2^{x+2b}(m+n)).
\]
\end{proof}

Next we show the claimed conditional lower bound on the running time for any algorithm for \CKC\ with $k=1$.
\begin{lem}\label{lem:k1LB}
Assuming the Exponential Time Hypothesis, there is no $2^{o(b)+f(x)}n^{O(1)}$ time algorithm for \CKC\ with $k=1$, for any function $f$.
\end{lem}
\begin{proof}
Since \CKC\ with $k=1$ and $x=0$ is equivalent to \VC, and assuming the Exponential Time Hypothesis, there is no $2^{o(k)}n^{O(1)}$ time algorithm for \VC\ \cite{lokshtanov2013lower},
we have that assuming the Exponential Time Hypothesis, there is no $2^{o(b)+f(x)}n^{O(1)}$ time algorithm for \CKC\ with $k=1$, where $f$ can be an arbitrary function.
\end{proof}

Before showing the correctness of Algorithm \ref{algo:1}, we first show in the following lemma why in line \ref{algst:too_big} set $Q$ should be bounded by $x+b$.
\begin{lem}
\label{lem:Q<b+x}
If $Q$ is of size more than $b+x$, then there is no solution containing whole $S$ and no vertex from $Q$.
\end{lem}
\begin{proof}
Suppose for contradiction that there is a set $Q$ of size at least $b+x+1$ and a solution $B$ such that $|B| \le b$, $S \subseteq B$ and $B \cap Q=\emptyset$. Let $v_1,v_2, \ldots, v_{r'}$ be the vertices of the set $S \cup Q$ in the order as they were added to the set by successive recursive calls. Moreover, if $B \setminus S$ is empty, then let $r=r'$. Otherwise, let $B \setminus S = \{v_{r'+1}, \ldots, v_r\}$, i.e., in both cases $\{v_1, \ldots, v_r\} = B \cup Q$.
Let $G'$ be the $1$-core of $G\setminus B$ and let $X = V(G') \cap Q$. Since $B$ is a solution, we know that $|X| \le x$. Our aim is to show that the number of edges lost by vertices of $Q \setminus X$ is larger than the number of edges incident to the vertices of $B$, which would be a contradiction. 

To this end, we construct an injective function $f$ that maps the vertices in $B$ to vertices of $Q \setminus X$. 
First let $v_i$ be the vertex in $B$ with the largest $i$. 
We let $f(i)$ be $\max \{j \mid v_j \in Q \setminus X\}$.
Now let $v_i$ be the vertex from $B$ with the largest $i$ such that $f(i)$ was not set yet and $v_{i'}$ be the vertex from $B$ with the least $i'$ such that $i'>i$. We set $f(i) = \max \{j \mid j < f(i') \wedge v_j \in Q \setminus X\}$. Since the set $Q \setminus X$ contains at least $b+1$ vertices, while $B$ contains at most $b$, this way we find a mapping for every vertex in $B$. Moreover, there remains at least one vertex in $Q \setminus X$ not being in the image of $f$, let us denote it $v_{q_0}$.

For $t \in \{1, \ldots, r\}$ let $G_t$ be the $1$-core of the graph $G \setminus (B \cap \{v_1, \ldots, v_{t-1}\})$. For $v \in V(G_t)$ let $\deg_t(v)$ be the degree of the vertex $v$ in $G_t$. 
By the way we selected $v_t$ we know that $\deg_t(v_t) \ge \deg_{t} (v_{t'})  \ge \deg_{t'} (v_{t'})$ for every $r' \ge t' > t$. Moreover, $\deg_t(v_t) \ge \deg_{t} (v_{t'}) \ge \deg_{t'} (v_{t'})$ for every $t \le r'< t'$.  
Note also that since $G_t$ is a $1$-core, we have $\deg_t(v_t) \ge 1$ for all $t$.

If for every vertex $v_i \in B$ we have $f(i)<i$, then $\deg_i(v_i)\leq \deg_{f(i)}(v_{f(i)})$
for every vertex $v_i$ in $B$. 
Let us count the number of edges of the form $\{v_i,v_j\}$ such that $v_i \in B$ and $v_j \in Q \setminus X$.
Since each edge incident on a vertex of $Q \setminus X$ must have the other endpoint in $B$, this number is at least $\sum_{v_j \in Q \setminus X} \deg_j (v_j)$. On the other hand, since edges towards $Q$ are always counted in $\deg_t(v)$, this number is at most $\sum_{v_i \in B} \deg_i (v_i)$. We have that 
\[
0 \le \sum_{v_i \in B} \deg_i (v_i) - \sum_{v_j \in Q \setminus X} \deg_j (v_j) \le - \deg_{q_0} (v_{q_0}) + \sum_{v_i \in B} (\deg_i (v_i) - \deg_{f(i)} (v_{f(i)})).
\]
On the other hand, since $\deg_i(v_i)\leq \deg_{f(i)}(v_{f(i)})$
for every vertex $v_i$ in $B$ we have that 
\[
- \deg_{q_0} (v_{q_0}) + \sum_{v_i \in B} (\deg_i (v_i) - \deg_{f(i)} (v_{f(i)})) \le - \deg_{q_0} (v_{q_0}) + \sum_{v_i \in B} 0 <0,
\]
which is a contradiction.

Otherwise, let $i_0$ be the largest $i$ such that $f(i) > i$ and $j_0=f(i_0)$. 
We consider the graph $G_{j_0}$. 
Now for every $i>i_0$ with $v_i \in B$ we have $i_0 < f(i_0)< f(i) < i$, thus $\deg_i(v_i)\leq \deg_{f(i)}(v_{f(i)})$. 
Let us count the number of edges of the form $\{v_i,v_j\}$ in $G_{j_0}$ such that $v_i \in B$ and $v_j \in Q \setminus X$. Since each edge of $G_{j_0}$ incident on a vertex of $(Q \setminus X) \cap V(G_{j_0})$ must have the other endpoint in $B \cap V(G_{j_0})$, this number is at least $\sum_{v_j \in (Q \setminus X)\cap V(G_{j_0})} \deg_j (v_j)$. On the other hand, since edges towards $Q$ are always counted in $\deg_t(v)$, this number is at most $\sum_{v_i \in B \cap V(G_{j_0})} \deg_i (v_i)$. We have that 
\begin{align*}
0 &\le \sum_{v_i \in B \cap V(G_{j_0})} \deg_i (v_i) - \sum_{v_j \in (Q \setminus X) \cap V(G_{j_0})} \deg_j (v_j) \\ &\le- \deg_{j_0} (v_{j_0}) + \sum_{v_i \in B\cap V(G_{j_0})} (\deg_i (v_i) - \deg_{f(i)} (v_{f(i)})) .
\end{align*}
On the other hand, since $\deg_i(v_i)\leq \deg_{f(i)}(v_{f(i)})$
for every vertex $v_i$ in $B\cap V(G_{j_0})$ we have that 
\[
- \deg_{j_0} (v_{j_0}) + \sum_{v_i \in B\cap V(G_{j_0})} (\deg_i (v_i) - \deg_{f(i)} (v_{f(i)})) \le - \deg_{j_0} (v_{j_0}) + \sum_{v_i \in B\cap V(G_{j_0})} 0 <0,
\]
which is again a contradiction.
\end{proof}
Now we have all necessary pieces to prove \autoref{thm:fpt_b+x_k=1}.
\begin{proof}[Proof for \autoref{thm:fpt_b+x_k=1}]
To show the correctness of the Algorithm~\ref{algo:1}, we first show that whenever the algorithm outputs a solution, then this solution is indeed correct. Then we show that whenever there exists a solution, the algorithm also finds a solution. 

$\Rightarrow:$ We show this part by induction on the recursion tree. If Algorithm~\ref{algo:1} returns $S$ as a solution in line \ref{algst:happy}, then it is of size at most $b$ since line \ref{algst:too_big} does not apply and the $1$-core of $G\setminus S$ is of size at most $x$. This constitutes the base case of the induction.

If the solution is obtained from recursive calls on lines~\ref{algst:rec_start}-\ref{algst:rec_end}, then we know that it is correct by induction hypothesis.

$\Leftarrow:$ Next, we show by induction on the recursion tree that if there is a solution then Algorithm~\ref{algo:1} returns a solution. In particular, we show that if there is a recursive call of Algorithm~\ref{algo:1} with sets $S$ and $Q$ such that there is a solution containing all of~$S$ and there is no solution containing the whole $S$ and any vertex from~$Q$, then the algorithm either directly outputs a solution or it invokes a recursive call with sets $S'$ and $Q'$ such that there is a solution containing all of~$S'$ and there is no solution containing the whole $S'$ and any vertex from~$Q'$.
Let $S$ and $Q$ be two input sets of a recursive call of Algorithm~\ref{algo:1} and let $B$ be a solution such that $S \subseteq B$ and $B \cap Q=\emptyset$. First we show that none of lines \ref{algst:too_big} and \ref{algst:stuck} applies. 
Line~\ref{algst:too_big} will not apply since we have $|S|\le |B| \le b$ and by \autoref{lem:Q<b+x} we also know that we have $|Q|\le b+x$.
If the $1$-core of $G[Q]$ is of size more than $x$, which is the case when line~\ref{algst:stuck} applies and line~\ref{algst:happy} does not apply, then no set $B'$ with $B' \cap Q= \emptyset$ can make the $1$-core of $G\setminus B'$ of size at most~$x$, which is a contradiction to $B$ being a solution with $B\cap Q=\emptyset$. If follows that the current recursive call of Algorithm~\ref{algo:1} does not output ``No Solution''.

If $B=S$, then line~\ref{algst:happy} applies and the algorithm outputs $B$. Otherwise, we know that any vertex $v$ is either in $B$ and hence $S\cup \{v\}\subseteq B$, or we have that $B\cap\{Q\cup \{v\}\}=\emptyset$. In particular, if the recursive call on line~\ref{algst:rec_start} does not return a solution, then, by induction hypothesis, there is no solution containing $S\cup \{v\}$, i.e., there is no solution containing $S$ and anything of $Q \cup \{v\}$ and the call on line~\ref{algst:rec_end} must return a solution by induction hypothesis. 
It follows that the algorithm invokes a recursive call with the desired properties in line~\ref{algst:rec_start} or in line~\ref{algst:rec_end}.

Since Algorithm~\ref{algo:1} starts with $S=Q=\emptyset$ we have that for any solution $B$ the conditions~$S \subseteq B$ and $B \cap Q=\emptyset$ are initially fulfilled. Furthermore, it follows from the complexity analysis in \autoref{lem:alg1runningtime} that the algorithm always terminates. Therefore the algorithm outputs a solution if one exists and hence is correct.

The running time bound for Algorithm~\ref{algo:1} follows from \autoref{lem:alg1runningtime} and the conditional running time lower bound for \CKC\ with $k=1$ follows from \autoref{lem:k1LB}.
\end{proof}

In the remainder of the section, we show how to adapt this algorithm for \CKC\ with~$k=2$.

\begin{theorem}
\label{thm:fpt_b+x_k=2}
\CKC\ with $k=2$ can be solved in $O(1.755^{x+4b}\cdot n)$ time.
Assuming the Exponential Time Hypothesis, there is no $2^{o(b)+f(x)}n^{O(1)}$ time algorithm for \CKC\ with $k=2$, for any function $f$.
\end{theorem}

The above theorem in particular yields an $O(9.487^b \cdot n)$ algorithm for \textsc{Feedback Vertex Set}.
For the proof we 
need the following lemma, which shows that, except for some specific connected components, we can limit the solution to contain vertices of degree at least three.

\begin{lem}
\label{lem:degree_at_least_3}
For any instance $(G,b,x,2)$ of \CKC\ with $k=2$, where $G$ is a $2$-core and does not contain a cycle as a connected component, if there is a solution $B$ for $(G,b,x,2)$, then there is also a solution $B'$ for $(G,b,x,2)$ which contains only vertices with degree larger than 2. 
\end{lem}
\begin{proof}
Let $v$ be any vertex with degree 2 in $B$. Let $v'$ be a vertex of degree at least 3 in $G$ such that there is a path $P$ between $v$ and $v'$ with all internal vertices of degree exactly 2 in $G$. Since no component of $G$ is a cycle, such a vertex must exist.
Let $B_1=B \cup v'$ and $B_2=B \cup v' \setminus v$.
We have the $2$-core of $G \setminus B_2$ is the same as the $2$-core of $G \setminus B_1$ and the $2$-core of $G \setminus B_1$ is a subset of the $2$-core of $G \setminus B$.
So the $2$-core of $G \setminus B_2$ is a subset of the $2$-core of $G \setminus B$, and hence no larger than $x$.
Therefore $B_2$ is also a solution for $(G,b,x,2)$.
Following the same way, we can replace all degree 2 vertices in $B$ with vertices which have degree larger than 2, and get a new solution $B'$ for $(G,b,x,2)$.
\end{proof}

\emph{Algorithm:} 
Our algorithm for $k=2$ is similar to Algorithm \ref{algo:1} with two main differences (see Algorithm~\ref{algo:2} for pseudocode). 
First $|Q| > b+x$ is replaced by $|Q| > 3b+x$.
Second when selecting the maximum degree vertex from $V(G') \setminus Q$, we need to make sure that this vertex has degree greater than 2.
Otherwise, either we can directly select vertices from $V(G') \setminus Q$ to break cycles in $G'$ and get a $2$-core of size at most $x$, or the algorithm rejects this branch.

\begin{algorithm}[t]
\caption{Algorithm for \CKC\ with $k=2$}\label{algo:2}
\textsc{SolveRec2}$(G,S,Q)$\\
\lIf{$|S| > b$ or $|Q| > 3b+x$}{\textbf{return} No solution\label{algst:2:S_too_big}} 
Let $G'$ be the $2$-core of $G \setminus S$\label{algst:2:2-core}.\\
\lIf{$|V(G')| \le x$}{\textbf{return} $S$\label{algst:2:happy}}
\If{$\max_{v \in V(G') \setminus Q} \deg(v) \le 3$ \label{algst:2:low_degrees}}{
Let $C_1, \ldots, C_r$ be the connected components of $G'$ not containing vertices of $Q$, ordered such that $|V(C_1)| \ge |V(C_2)| \ge \ldots \ge |V(C_r)|$\label{algst:2:components}\;
Let $r'\leftarrow \min \{r, b -|S|\}$\label{algst:2:budget}\;
\If{$|V(G')|-\sum_{i=1}^{r'}|V(C_i)| \le x$\label{algst:2:check_size}}{
\lFor{$i=1, \ldots, r'$}{Select an arbitrary vertex from $C_i$ and add it to $S$}
\textbf{return} $S$\label{algst:2:greedy_end}
}
\lElse{\textbf{return} No solution\label{algst:2:no_greedy}}
}
Let $v$ be the vertex with the highest degree in $G'$ which is not in $Q$\label{algst:2:select}\;
$T \leftarrow$ \textsc{SolveRec2}$(G,S \cup \{v\},Q)$ \label{algst:2:rec_start}\;
\lIf{$T\neq \text{\normalfont{No solution}}$}{\textbf{return} $T$}
\lElse{\textbf{return} \textsc{SolveRec2}$(G,S,Q \cup \{v\})$ \label{algst:2:rec_end}} 

\end{algorithm}

We start by showing that Algorithm~\ref{algo:2} has the claimed running time.
\begin{lem}\label{lem:algo2runningtime}
Algorithm~\ref{algo:2} runs in $O(1.755^{x+4b}\cdot n)$ time.
\end{lem} 
\begin{proof}
We first show by induction on the size of $S \cup Q$ starting with the largest size achieved that a call with $S$ and $Q$ results in at most $2(\binom{4b+x+2-|S|-|Q|}{b+1-|S|})-1$ calls to the function in total.
Indeed, the size of $|Q|$ never exceeds $3b+x+1$ and the size of $|S|$ never exceeds $b+1$ since the sizes grow by one and if they achieve the bound, then line~\ref{algst:2:S_too_big} applies. Hence, if any of the lines~\ref{algst:2:S_too_big}--\ref{algst:2:no_greedy} applies, then we have only one call and $|Q| \le 3b+x+1$ and $|S|\le b+1$ implies $2(\binom{4b+x+2-|S|-|Q|}{b+1-|S|})-1 \ge 1$, making the basic cases. If the call makes recursive calls, then in one of them $S$ is one larger than in the current one, and if the second one is made, then $Q$ is one larger in it. Hence the number of calls is at most $1+ 2(\binom{4b+x+2-|S|-|Q|-1}{b+1-|S|-1})-1+2(\binom{4b+x+2-|S|-|Q|-1}{b+1-|S|})-1 \le 2(\binom{4b+x+2-|S|-|Q|}{b+1-|S|})-1$, finishing the induction.

Since we call the algorithm with sets $S$ and $Q$ empty, it follows that the total number of calls is at most $2(\binom{4b+x+2}{b+1}) \le 2\binom{4b+x+4}{b+1}$. Using, e.g., the lemma of Fomin et al.~\cite[Lemma 10]{FominKPPV14} (or Stirling's approximation) one can show that $\binom{z}{\frac{1}{4}z} = O(1.7549^z)$. Hence, $\binom{4b+x+4}{b+1} \le \binom{4b+x+4}{\frac{1}{4}(4b+x+4)} =O(1.7549^{4b+x+4})$. So we have that the number of recursive calls is at most $O(1.7549^{4b+x})$.

Now we analyze the time complexity of a single call.
In line \ref{algst:2:2-core}, the $2$-core $G'$ can be found in $O(m)$ time by recursively removing vertices with degree less than 2 in $G \setminus S$~\cite{batagelj2003m}.
In line \ref{algst:2:components}, sorting these circles according to their sizes can be done in $O(n)$ time using counting sort. 
Thus except for Line \ref{algst:2:rec_start} and \ref{algst:2:rec_end}, all steps can be done in $O(m+n)$. 
Since there are at most $O(bn+x^2)$ edges or we are facing a no-instance, we have that the time complexity of the algorithm is $O(1.7549^{4b+x}\cdot b^{O(1)}x^{O(1)}n)$, which is $O(1.755^{4b+x}\cdot n)$ as claimed.
\end{proof}
Next we show the claimed conditional lower bound on the running time for any algorithm for \CKC\ with $k=2$.
\begin{lem}\label{lem:k2lowerbound}
Assuming the Exponential Time Hypothesis, there is no $2^{o(b)+f(x)}n^{O(1)}$ time algorithm for \CKC\ with $k=2$, for any function $f$.
\end{lem}
\begin{proof}
Since \CKC\ with $k=2$ and $x=0$ is equivalent to \FVS, and assuming the Exponential Time Hypothesis, there is no $2^{o(k)}n^{O(1)}$ time algorithm for \FVS~\cite{lokshtanov2013lower},
we have that assuming the Exponential Time Hypothesis, there is no $2^{o(b)+f(x)}n^{O(1)}$ time algorithm for \CKC\ with $k=2$, where $f$ is an arbitrary function.
\end{proof}

Before showing the correctness of Algorithm \ref{algo:2}, we prove the following lemmata which will be helpful in the correctness proof.
The next lemma helps to show that line \ref{algst:2:no_greedy} is correct.
\begin{lem}
\label{lem:line_14_correct}
If line \ref{algst:2:no_greedy} of Algorithm \ref{algo:2} applies, and there is no solution containing whole $S$ and no vertex from $Q$, then there is no solution containing the whole $S$.
\end{lem}
\begin{proof}
Suppose for contradiction that $B$ is a solution containing $S$.
Note that in this case $B'= B\setminus S$ is a solution for $(G',b -|S|, x, 2)$.
Let $D$ be the graph formed by the union of connected components of $G'$ containing vertices of degree at least 3 and $C_0$ be the graph formed by the union of connected components of $G'$ which are cycles and contain vertices of $Q$, that is, $V(D) \cup V(C_0) = V(G') \setminus \bigcup_{i=1}^{r}V(C_i)$, where $C_i$ is as stated in the algorithm.

If $B^D= B\cap V(D)$ is nonempty and $x'$ is the size of the $2$-core of $D \setminus B^D$, then $B^D$ is a solution to the instance $(D,|B^D|,x',2)$. Hence, by Lemma~\ref{lem:degree_at_least_3}, there is another solution 
$B^D_3$ to the instance $(D,|B^D|,x',2)$ which contains only vertices of degree at least 3. However all these vertices are in $Q$, and $(B \setminus B^D) \cup B^D_3$ is a solution to $(G,b,x,2)$ containing whole $S$ and vertices of $Q$, contradicting our assumption.
Hence $B^D$ is empty. 

If $B^C= B\cap V(C_0)$ is nonempty, then let $y$ be any vertex in $B^C$ and $C_y$ the connected component of $C_0$ containing $y$. By the definition of $C_0$ component $C_y$ contains a vertex of $Q$, let us denote it $y'$.
Let $B'= (B \setminus \{y\}) \cup \{y'\}$. The $2$-cores of $G \setminus B$ and $G \setminus B$ are the same, since $C_y$ is a cycle. 
Thus $B'$ is a solution to $(G,b,x,2)$ containing whole $S$ and vertices of $Q$, contradicting our assumption.
Hence also $B^C$ is empty. 

The components of $G'$ neither in $C_0$ nor in $D$ are cycles since $G'$ is a $2$-core, they contain no vertices of $Q$, and there are no other vertices of degree at least 3. Since $B$ contains at most $r'= \min \{r, b -|S|\}$ vertices out of these components, it can destroy at most $r'$ of these cycles. Since $|V(C_1)| \ge |V(C_2)| \ge \ldots \ge |V(C_r)|$, this decreases the size of the $2$-core by at most $\sum_{i=1}^{r'}|V(C_i)|$. Thus, if $|V(G')|-\sum_{i=1}^{r'}|V(C_i)| > x$, then there is no solution containing whole $S$.
\end{proof}

Next, we show that the second part of line \ref{algst:2:S_too_big} of Algorithm \ref{algo:2} is correct.
\begin{lem}
\label{lem:line_3_correct}
If $Q$ is of size more than $3b+x$, then there is no solution containing whole $S$ and no vertex from $Q$.
\end{lem}

\begin{proof}
Suppose for contradiction that there is a set $Q$ of size at least $3b+x+1$ and a solution $B$ such that $|B| \le b$, $S \subseteq B$ and $B \cap Q=\emptyset$. Let $v_1,v_2, \ldots, v_{r'}$ be the vertices of the set $S \cup Q$ in the order as they were added to the set by successive recursive calls. Moreover, if $B \setminus S$ is empty, then let $r=r'$. Otherwise, let $B \setminus S = \{v_{r'+1}, \ldots, v_r\}$, i.e., in both cases $\{v_1, \ldots, v_r\} = B \cup Q$.
Let $G'$ be the $2$-core of $G\setminus B$ and let $X = V(G') \cap Q$. Since $B$ is a solution, we know that $|X| \le x$.

We construct a function $f$ that maps every vertex of $B$ to a set of three consecutive vertices of $Q \setminus X$ (see also Figure \ref{fig:function}). 
First let $v_i$ be the vertex in $B$ with the largest $i$. 
We let $f(i)$ be the set $\{j_1,j_2,j_3\}$, where $j_1=\max\{j \mid v_j \in Q \setminus X\}$, $j_2=\max\{j<j_1 \mid v_j \in Q \setminus X\}$, and $j_3=\max\{j<j_2 \mid v_j \in Q \setminus X\}$.
Now let $v_i$ be the vertex from $B$ with the largest $i$ such that $f(i)$ was not set yet and $v_{i'}$ be the vertex from $B$ with the least $i'$ such that $i'>i$. We set $f(i) = \{j_1,j_2,j_3\}$, where $j_1=\max\{j< \min\{k \in f(i')\} \mid v_j \in Q \setminus X\}$, $j_2=\max\{j<j_1 \mid v_j \in Q \setminus X\}$, and $j_3=\max\{j<j_2 \mid v_j \in Q \setminus X\}$. Since the set $Q \setminus X$ contains at least $3b+1$ vertices, while $B$ contains at most $b$, this way we find a mapping for every vertex in $B$, keeping the images of different vertices disjoint. Moreover, denote $p=|Q \setminus X|-3|B| \geq 1$. There remains $p$ vertices in $Q \setminus X$ not being in the union of images of $f$, let us denote them $v_{q_1}, \dots, v_{q_{p}}$.

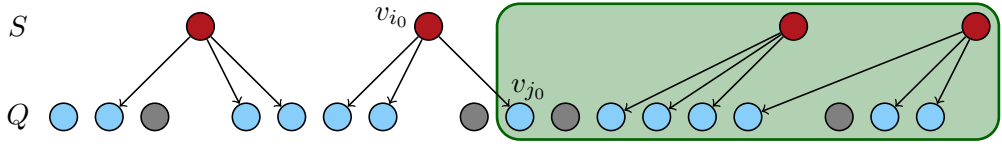
\begin{figure}[t]
\begin{center}
\begin{tikzpicture}[scale=.6, line width=.6pt]

\node (S) at (4,2) {$S$};
\node (S) at (4,0) {$Q$};

\node[draw,circle,fill=ourred] (s2) at (8,2) {};
\node[draw,circle,fill=ourred] (s3) at (13,2) {};
\node[draw,circle,fill=ourred] (s4) at (21,2) {};
\node[draw,circle,fill=ourred] (s5) at (25,2) {};

\node[draw,circle,fill=ourblue] (q3) at (5,0) {};
\node[draw,circle,fill=ourblue] (q4) at (6,0) {};
\node[draw,circle,fill=ourblue] (q5) at (9,0) {};
\node[draw,circle,fill=ourblue] (q6) at (10,0) {};
\node[draw,circle,fill=ourblue] (q7) at (11,0) {};
\node[draw,circle,fill=ourblue] (q8) at (12,0) {};
\node[draw,circle,fill=ourblue] (q9) at (15,0) {};
\node[draw,circle,fill=ourblue] (q10) at (17,0) {};
\node[draw,circle,fill=ourblue] (q11) at (18,0) {};
\node[draw,circle,fill=ourblue] (q12) at (19,0) {};
\node[draw,circle,fill=ourblue] (q13) at (20,0) {};
\node[draw,circle,fill=ourblue] (q14) at (23,0) {};
\node[draw,circle,fill=ourblue] (q15) at (24,0) {};

\node[draw,circle,fill=gray] (x2) at (7,0) {};
\node[draw,circle,fill=gray] (x1) at (14,0) {};
\node[draw,circle,fill=gray] (x1) at (16,0) {};
\node[draw,circle,fill=gray] (x2) at (22,0) {};

\draw[->] (s2) -- (q4);
\draw[->] (s2) -- (q5);
\draw[->] (s2) -- (q6);
\draw[->] (s3) -- (q7);
\draw[->] (s3) -- (q8);
\draw[->] (s3) -- (q9);
\draw[->] (s4) -- (q10);
\draw[->] (s4) -- (q11);
\draw[->] (s4) -- (q12);
\draw[->] (s5) -- (q13);
\draw[->] (s5) -- (q14);
\draw[->] (s5) -- (q15);

\node (i0) at (12.2,2.2) {$v_{i_0}$};
\node (i0) at (15.2,0.6) {$v_{j_0}$};

\begin{pgfonlayer}{bg}

\draw[rounded corners=3mm, fill=ourgreen!30] (14.5,2.5) -- (14.5,-.5) -- (25.5,-.5) -- (25.5,2.5) -- cycle;
\draw[rounded corners=3mm, line width=1pt,draw=ourgreen] (14.5,2.5) -- (14.5,-.5) -- (25.5,-.5) -- (25.5,2.5) -- cycle;

\end{pgfonlayer}

\end{tikzpicture}
\end{center}
\caption{Illustration of function $f$. Vertices in $Q$ are separated into two parts: gray vertices from~$X$ and blue vertices from $Q'=Q \setminus X$.
Every vertex in $S$ is mapped to a set of three consecutive blue vertices in $Q'$ from the right to the left. Graph $G_{j_0}$ with the property that $i > \max\{j \mid j \in f(i)\}$ for every $i$ with $v_i \in B\cap V(G_{j_0})$ is contained in the green box.
}
\label{fig:function}
\end{figure}

For $t \in \{1, \ldots, r\}$ let $G_t$ be the $2$-core of the graph $G \setminus (B \cap \{v_1, \ldots, v_t\})$. For $v \in V(G_t)$ let $\deg_t(v)$ be the degree of the vertex $v$ in $G_t$. 
By the way we selected $v_t$ we again know that $ \deg_t(v_t) \ge \deg_{t} (v_{t'}) \ge \deg_{t'} (v_{t'}) $ for every $t < t' \le r'$ and $\deg_t(v_t) \ge \deg_{t} (v_{t'}) \ge \deg_{t'} (v_{t'})$ for every $t \le r' < t' $. 
Note also that $\deg_t(v_t) \ge 3$ for all $v_t \in Q$.

If for every vertex $v_i \in B$ we have $i > \max\{j \mid j \in f(i)\}$, then $\deg_i(v_i)\leq \deg_j(v_{j})$ for vertex $v_i$ in $B$ and every $j \in f(i)$.
Together with $\deg_t(v_t) \ge 3$ for all $v_t \in Q$, we have $\deg_i(v_i) - \sum_{j \in f(i)}\deg_t(v_{j})+6 \le 0$
for every vertex $v_i$ in $B$. 
Let $V=B \uplus Q' \uplus X \uplus Y$, where $Q'=Q \setminus X$ and $Y$ is the set of collapsed vertices not contained in $Q$.
Now we transform graph $G$ into a partial directed graph by considering the collapsing process (see also \autoref{fig:directed_graph}).
More precisely, for every edge in $G$, except for edges which have two endpoints in $X$, we will assign it a direction.
To this end, we just need to give an order of vertices in $V(G) \setminus X$.
We define this order based on the time vertices collapse.
It may happens that several vertices in~$Q'$ or $Y$ collapse at the same time.
For this situation, we just order these vertices according to an arbitrary but fixed order.
Since $k=2$, every collapsed vertex in $Q' \cup Y$ has at most one outgoing edge.
Let us consider the number, denoted by $N$, of edges of the form $\overrightarrow{v_iv_j}$ such that the head $v_j$ is in $Q'$.
Since every vertex in $Q'$ has at most one outgoing edge,
we have $N \ge \sum_{v_j \in Q'} \deg_j(v_j)-|Q'|$.

\begin{figure}[t]
\begin{center}
\begin{tikzpicture}[line width=1pt, scale=1] 
 
    \draw (-4,0) circle (0.7cm);
    \draw (0,0) circle (0.7cm);
    \draw (1,0) circle (0.7cm);
    \draw (-2,-2) circle (0.7cm);
    
    \node (B) at (-4,0) {$B$};
    \node (Q') at (0,0) {$Q'$};
    \node (X) at (1,0) {$X$};
    \node (Y) at (-2,-2) {$Y$};
    
    \draw[->] (B) -- (Q');
    \draw[->] (B) -- (Y);
    \draw[->] (Y) -- (Q');
    \draw[->] (Q') -- (Y);
    
    \node (BQ') at (-2,0.3) {$N_{\overrightarrow{BQ'}}$};
    \node (BY) at (-3.4,-1.2) {$N_{\overrightarrow{BY}}$};
    \node (YQ') at (-1.6,-0.8) {$N_{\overrightarrow{YQ'}}$};
    \node (Q'Y) at (-0.4,-1.2) {$N_{\overrightarrow{Q'Y}}$};
    
\end{tikzpicture}
\end{center}
\caption{Illustration of the partial directed graph when considering the collapsing process. The set $B$ contains the deleted vertices and $X$ is the set of vertices remaining in the $2$-core of $G \setminus B$. The set $Q'$ contains vertices the algorithm has decided not to delete but eventually collapse and $Y$ is the set of other collapsed vertices.}
\label{fig:directed_graph}
\end{figure}
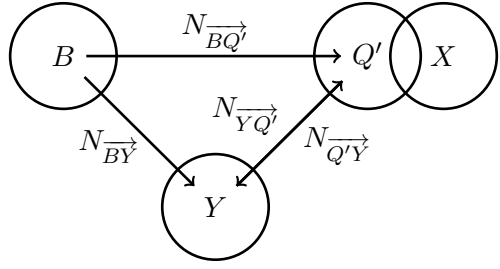

On the other hand, let 
\begin{compactitem}
\item $N_{\overrightarrow{BQ'}}$ be the number of edges going from $B$ to $Q'$;
\item $N_{\overrightarrow{BY}}$ be the number of edges going from $B$ to $Y$;
\item $N_{\overrightarrow{YQ'}}$ be the number of edges going from $Y$ to $Q'$;
\item $N_{\overrightarrow{Q'Y}}$ be the number of edges going from $Q'$ to $Y$.
\end{compactitem}

We claim that 
\begin{equation}
N_{\overrightarrow{YQ'}} \le N_{\overrightarrow{BY}}+N_{\overrightarrow{Q'Y}}.\label{claim1}
\end{equation}
Denote the number of edges in $G[Y]$ by $n_Y$.
Since every vertex in $Y$ has at least one incoming edge but at most one outgoing edge, 
we have $\sum_{v \in Y} \deg^-(v) \le \sum_{v \in Y} \deg^+(v)$,
where $\deg^{+}(v)$ ($\deg^{-}(v)$) is the number of incoming (outgoing) edges of vertex $v$ in $G$, respectively.
This means $N_{\overrightarrow{YQ'}}+ n_Y \le N_{\overrightarrow{BY}}+N_{\overrightarrow{Q'Y}} +n_Y$,
therefore, $N_{\overrightarrow{YQ'}} \le N_{\overrightarrow{BY}}+N_{\overrightarrow{Q'Y}}$, finishing the proof of the claim.

Since the edges which have their heads in $Q'$ have their tails from $B \cup Y \cup Q'$,
\[N \le N_{\overrightarrow{BQ'}}+ N_{\overrightarrow{YQ'}} + (|Q'|-N_{\overrightarrow{Q'Y}})
\le N_{\overrightarrow{BQ'}}+ |Q'| + N_{\overrightarrow{BY}}
\le \sum_{v_i \in B} \deg_i (v_i) + |Q'|.
\]
Then we have that 
\begin{align*}
0 & \le \sum_{v_i \in B} \deg_i (v_i)+|Q'| - \Big(\sum_{v_j \in Q'} \deg_j (v_j)-|Q'| \Big) \\
& = \sum_{v_i \in B} \deg_i (v_i) - \sum_{v_j \in Q'} (\deg_j (v_j)-2) \\
& \le  \sum_{v_i \in B} \deg_i (v_i) - \Big(\sum_{v_i \in B}\sum_{j \in f(i)} (\deg_j (v_j)-2) + \sum_{i=1}^p (\deg_{q_i}(v_{q_i})-2 ) \Big) \\
& < \sum_{v_i \in B} \deg_i (v_i) - \sum_{v_i \in B}\sum_{j \in f(i)} (\deg_j (v_j)-2) \\
& = \sum_{v_i \in B} \Big(\deg_i (v_i) - \sum_{j \in f(i)} \deg_j (v_j)+6 \Big) \\
& \le 0,
\end{align*}
which is a contradiction.

Otherwise, let $i_0$ be the largest $i$ such that $i < \max\{j \mid j \in f(i)\}$ and $j_0=\max\{j \mid j \in f(i_0)\}$.
Now we consider the graph $G_{j_0}$ and the function $f$ restricted on $B\cap V(G_{j_0})$ (see also Figure \ref{fig:function}). 
First note that in $G_{j_0}$ vertex $v_{j_0}$ is the only vertex which is not contained in any image set of $f$ restricted on $B\cap V(G_{j_0})$.
Moreover, for every $i$ with $v_i \in B\cap V(G_{j_0})$ we have $i > \max\{j \mid j \in f(i)\}$.
Hence, $\deg_i(v_i)\leq \deg_j(v_{j})$ for vertex $v_i$ in $B\cap V(G_{j_0})$  and every $j \in f(i)$.
Together with $\deg_t(v_t) \ge 3$ for all $v_t \in Q$, we have $\deg_i(v_i) - \sum_{j \in f(i)}\deg_t(v_{j})+6 \le 0$
for every vertex $v_i$ in $B\cap V(G_{j_0})$. 

Now similar to the first situation, we first transform graph $G_{j_0}$ into a partial directed graph by considering the collapsing process.
Let $V(G_{j_0})=B_0 \uplus {Q_0}' \uplus X_0 \uplus Y_0$, where $B_0=B \cap V(G_{j_0})$, ${Q_0}'=(Q \setminus X) \cap V(G_{j_0})$, $X_0=X \cap V(G_{j_0})$ and $Y_0$ is the set of collapsed vertices in $V(G_{j_0})$ not contained in $Q$. 
Then consider the number, denoted by $N_0$, of edges of the form $\overrightarrow{v_iv_j}$ such that the head $v_j$ is in ${Q_0}'$.
Since every vertex in ${Q_0}'$ has at most one outgoing edge,
we have $N_0 \ge \sum_{v_j \in {Q_0}'} \deg_j(v_j)-|{Q_0}'|$.
On the other hand, let 
\begin{compactitem}
\item $N_{\overrightarrow{B_0{Q_0}'}}$ be the number of edges going from $B_0$ to ${Q_0}'$;
\item $N_{\overrightarrow{B_0Y_0}}$ be the number of edges going from $B_0$ to $Y_0$;
\item $N_{\overrightarrow{Y_0{Q_0}'}}$ be the number of edges going from $Y_0$ to ${Q_0}'$;
\item $N_{\overrightarrow{{Q_0}'Y_0}}$ be the number of edges going from ${Q_0}'$ to $Y_0$.
\end{compactitem}
Similar to \autoref{claim1}, we have $N_{\overrightarrow{Y_0{Q_0}'}} \le N_{\overrightarrow{B_0Y_0}}+N_{\overrightarrow{{Q_0}'Y_0}}$ and get the upper bound for $N_0$:

\[N \le N_{\overrightarrow{B_0{Q_0}'}}+ N_{\overrightarrow{Y_0{Q_0}'}} + (|{Q_0}'|-N_{\overrightarrow{{Q_0}'Y_0}})
\le N_{\overrightarrow{B_0{Q_0}'}}+ |{Q_0}'| + N_{\overrightarrow{B_0Y_0}}
\le \sum_{v_i \in B_0} \deg_i (v_i) + |{Q_0}'|.
\]
Then we have that
\begin{align*}
0 & \le \sum_{v_i \in B_0} \deg_i (v_i)+|{Q_0}'| - \Big(\sum_{v_j \in {Q_0}'} \deg_j (v_j)-|{Q_0}'| \Big) \\
& = \sum_{v_i \in B_0} \deg_i (v_i) - \sum_{v_j \in {Q_0}'} (\deg_j (v_j)-2) \\
& \le  \sum_{v_i \in B_0} \deg_i (v_i) - \Big(\sum_{v_i \in B_0}\sum_{j \in f(i)} (\deg_j (v_j)-2) + \deg_{j_0}(v_{j_0})-2 \Big) \\
& < \sum_{v_i \in B_0} \deg_i (v_i) - \sum_{v_i \in B_0}\sum_{j \in f(i)} (\deg_j (v_j)-2) \\
& = \sum_{v_i \in B_0} \Big(\deg_i (v_i) - \sum_{j \in f(i)} \deg_j (v_j)+6 \Big) \\
& \le 0, 
\end{align*}
which is a contradiction. 
\end{proof}
Now we have all necessary pieces for the proof of \autoref{thm:fpt_b+x_k=2}.

\begin{proof}[Proof for \autoref{thm:fpt_b+x_k=2}]
To show the correctness of the algorithm, it is again enough to show that for any pair of $S$ and $Q$ appearing in the recursive process, the set returned by the recursive function is a solution that contains whole $S$ and no vertex from $Q$ and if there is a solution containing whole $S$ and there is no solution containing the whole $S$ and any vertex from $Q$, then the function returns such a solution.
For simplicity we assume without loss of generality that the input graph is a $2$-core.

$\Leftarrow:$ If the function returns $S$ as a solution on line \ref{algst:2:happy}, then it is of size at most $b$ since line \ref{algst:2:S_too_big} does not apply and the $2$-core of $G\setminus S$ is of size at most $x$. Hence it is obviously a solution and it is fulfilling the constraints.
If the solution is obtained from recursive calls on lines \ref{algst:2:rec_start}-\ref{algst:2:rec_end}, then it is returned without modification and $S$ is subset of the solution, since $S$ or its superset was passed to the recursive calls. Similarly, it contains no vertex from $Q$, since $Q$ or its superset was passed to the recursive calls.

If the function returns set $S'$ as a solution on line \ref{algst:2:greedy_end}, then for each $i \in \{1, \ldots, r'\}$ set $S'$ contains a vertex of the cycle $C_i$. Hence the $2$-core of $G \setminus S'$ does not contain the cycle $C_i$. Since the $2$-core $G''$ of $G \setminus S'$ differs from the $2$-core $G'$ of $G \setminus S$ exactly in these missing cycle components, we have $V(G'') = |V(G')|-\sum_{i=1}^{r'}|V(C_i)|$. Since this is at most $x$ by line \ref{algst:2:check_size}, line \ref{algst:2:S_too_big} does not apply, and $Q \cap \bigcup_{i=1}^{r'}V(C_i)=\emptyset$, $S$ is a solution fulfilling the constraints.  Hence, if the function returns a solution, then the answer is correct.

$\Rightarrow:$ Now we show by induction on $|Q \cup S|$ for every pair of $S$ and $Q$ appearing in the recursive process, starting from the largest $|Q \cup S|$ achieved, that if there is a solution containing whole $S$ and there is no solution containing the whole $S$ and any vertex from $Q$, then 
the algorithm returns such a solution. 

If $B$ is a solution such that $S \subseteq B$ and $B \cap Q=\emptyset$, then line \ref{algst:2:S_too_big} will not apply since $|B| \le b$ and because of \autoref{lem:line_3_correct}, line~\ref{algst:2:no_greedy} does not apply according to \autoref{lem:line_14_correct}. 

Let $G'$ be the $2$-core of $G \setminus S$. If $B$ is a solution, then $S \cup ((B \setminus S) \cap V(G'))$ is also a solution, so we can assume that $B \setminus (S \cup V(G'))=\emptyset$. 
If $B=S$, then line \ref{algst:happy} applies. If there are no vertices of degree at least 3 which are not in $Q$, then $B$ contains no vertices of the components containing vertices of $Q$, as we have shown above.
Hence $B$ can decrease the size of the $2$-core compared to $G'$ by at most $\sum_{i=1}^{r'}|V(C_i)|$, where $r'$ is as in the algorithm. As $B$ is a solution, this implies $|V(G')|-\sum_{i=1}^{r'}|V(C_i)| \le x$ and a solution containing $S$ is returned on line \ref{algst:2:greedy_end}.
This finishes the proof for the base cases of the induction.

Now suppose that the claim already holds for all calls with larger $|Q \cup S|$ and $v$ is the vertex selected by the algorithm on line \ref{algst:2:select}.
If there is a solution containing $S \cup \{v\}$ (in particular if $v \in B$), then the call \textsc{SolveRec2}$(G,S \cup \{v\},Q)$ must return a solution by induction hypothesis and otherwise there is no solution containing $S$ and anything of $Q \cup \{v\}$ and the call \textsc{SolveRec2}$(G,S,Q \cup \{v\})$ must return a solution by induction hypothesis. Thus the algorithm works correctly.

The running time bound for Algorithm~\ref{algo:2} follows from \autoref{lem:algo2runningtime} and the conditional running time lower bound for \CKC\ with $k=1$ follows from \autoref{lem:k2lowerbound}.
\end{proof}

\section{Structural Graph Parameters}\label{sec:structural}
In this section, we investigate the parameterized complexity of \CKC\ with respect to several structural parameters of the input graph. \cref{cor:whardness} already implies hardness for constant values of several structural graph parameters. We expand this picture by observing that the problem remains \NP-hard on graphs with a dominating set of size one and by showing that the problem is \W{1}-hard when parameterized by the combination of~$b$ and the clique cover number of the input graph. 
On the positive side, we show that the problem is in \FPT when parameterized by the treewidth of the input graph or the clique-width of the input graph and~$k$ combined with either~$b$, $x$, $n-x$, or $n-b$. Lastly, we show that the problem presumably does not admit a polynomial kernel when parameterized by the combination of $b$ and the vertex cover number of the input graph, or by the combination of~$b$, $k$, and the bandwidth of the input graph.

We start with an easy observation that we will make use of in most of the hardness results in this section.
\begin{obs}\label{obs:high_deg}
If $(G,b,x,k)$ is an instance of \CKC\ and vertex $v$ is a part of the $(k+b)$-core of $G$, and $S \subseteq V$ is of size at most $b$, then either $v \in S$ or $v$ is part of the $k$-core of $G\setminus S$.
\end{obs}
\begin{proof}
 Let $C$ be the $(k+b)$-core of $G$. In $C\setminus S$ the degree of each vertex is at least $k+b-b$, hence $C \setminus S$ is a subgraph of the $k$-core of $G\setminus S$.
\end{proof}

The following observation yields that we can reduce the size of a dominating set of any instance of \CKC\ to one by introducing a universal vertex. Note that, for example, this only increases the degeneracy by one.
\begin{obs}\label{obs:dominatingset}
Let $(G,b,x,k)$ be an instance of \CKC\ and $G'$ be the graph obtained from $G$ by adding a universal vertex, then $(G',b+1,x,k)$ is an equivalent instance of \CKC.
\end{obs} 
\begin{proof}
 Let $(G,b,x,k)$ be an instance of \CKC\ and let $(G',b',x,k)$ be the instance formed by a graph $G'$ which is obtained from $G$ by adding an universal vertex $u$, $b'=b+1$, and $x$ and $k$ from the original instance. We claim that the instances are equivalent.
 
 First if $S$ is a solution for $(G,b,x,k)$, then $S \cup \{u\}$ is a solution for $(G',b',x,k)$, as $G \setminus S = G' \setminus S'$. Second, let $S'$ be a solution for $(G',b',x,k)$. If $S'$ contains $u$, then $S' \setminus \{u\}$ is a solution for $G$. Now suppose that $S'$ does not contain $u$ and let $v$ be an arbitrary vertex of $S'$. We claim that $S''=(S' \setminus \{v\}) \cup \{u\}$ is also a solution to $(G',b',x,k)$, since $G' \setminus S''$ is isomorphic to a subgraph of $G'\setminus S'$. Indeed, consider the bijection $\varphi$, which maps each vertex of $V(G) \setminus S'$ to itself and $v$ to $u$. To show that it is an isomorhism, it is enough to consider edges incident on $v$, however, as there is an edge between $u$ and every vertex of $V(G) \setminus S'$, these definitely map to edges. Hence the $k$-core of $G' \setminus S''$ is at most as large as the $k$-core of $G'\setminus S'$ and indeed $S''$ is a solution to $(G',b',x,k)$. Now the equivalence of the instance follows from the case where $u$ is in $S'$.
\end{proof}

Considering a larger parameter than e.g.\ the size of the dominating set, namely the clique cover number\footnote{The clique cover number of a graph $G$ is the minimum number of induced cliques such that their union contains all vertices of $G$.}, we can show \W{1}-hardness, even in combination with $b$. We do this by providing a parameterized reduction from \textsc{Multicolored Clique} parameterized by the solution size.

\tikzstyle{bigclique}=[ellipse,draw=black,minimum size=15pt,inner sep=2pt,fill=lightgray, line width=1.8pt]
\tikzstyle{bigclique2}=[ellipse,draw=black,minimum size=15pt,inner sep=2pt,fill=lightgray]
\tikzstyle{vclique}=[ellipse,draw=black,minimum size=1.2cm,inner sep=2pt]
\tikzstyle{eclique}=[rectangle,draw=black,minimum size=1.2cm,inner sep=2pt]
\tikzstyle{thickgray}=[line width=1.8pt, color=gray]
\tikzstyle{smallv}=[draw,shape=circle,fill,inner sep=2pt]

\tikzstyle{sclique}=[circle,draw=black,minimum size=1.2cm,inner sep=2pt,fill=lightgray, line width=1.8pt]
\tikzstyle{vset}=[ellipse,draw=black,minimum size=15pt,inner sep=2pt,fill=white]
\tikzstyle{eset}=[ellipse,draw=black,minimum size=15pt,inner sep=2pt,minimum width=2cm,minimum height=1cm, line width=1pt,fill=white]

\begin{prop}
\label{thm:W1-hard_wrt_clique_cover}
\CKC\ is \W{1}-hard when parameterized by the combination of~$b$ and the clique cover number of the input graph.
\end{prop}
\begin{figure}[t]
\begin{center}
\begin{tikzpicture}[scale=.91, line width=.6pt]
\node[bigclique2,minimum width=3cm, minimum height=1cm] (C1) at (0,3) {$C$};
\node[vclique,minimum width=1cm, minimum height=3cm] (V1) at (-4,0) {};
\node[vclique,minimum width=1cm, minimum height=3cm] (V2) at (0,0) {};
\node[vclique,minimum width=1cm, minimum height=3cm] (Vs) at (4,0) {};
\node[eclique,minimum width=1cm, minimum height=3cm] (E12) at (-2,0) {};
\node (dots) at (2,0) {$\dots$};
\node (e12) at (-2,-2) {$E_{1,2}$};
\node (v1) at (-4,-2) {$V_1$};
\node (v2) at (0,-2) {$V_2$};
\node (vs) at (4,-2) {$V_s$};
\node[bigclique2,minimum width=3cm, minimum height=1cm] (C2) at (0,-3) {$C$};
\path (V1) edge[in=180, out=90,thickgray] (C1);
\path (V2) edge[in=270, out=90,thickgray] (C1);
\path (Vs) edge[in=0, out=90,thickgray] (C1);

\node(d1) at (-2,1.2) {$\ldots$};
\node[smallv,fill=lipicsyellow] (ev1) at (-2,0.7) {};
\node[smallv,fill=lipicsyellow] (ev2) at (-2,0.2) {};
\node(d2) at (-2,-.3) {$\ldots$};
\node[smallv, fill=ourred] (v11) at (-4,0.45) {};
\node[smallv, fill=ourblue] (v21) at (0,0.45) {};
\draw (v11) -- (ev1);
\draw (v11) -- (ev2);
\draw (v21) -- (ev1);
\draw (v21) -- (ev2);

\path (-2.55,-.5) edge (-1.45,-.5);

\node(d2) at (-2,-.8) {$\ldots$};
\node[smallv] (ev3) at (-2,-1.2) {};
\path (ev3) edge[in=120, out=300,thickgray] (C2);

\path (C1) edge[in=20, out=340,thickgray] (C2);

\end{tikzpicture}
\end{center}
\caption{Illustration of the reduction from \MCC to \CKC with~$E_{1,2}$ highlighted.
Every gray edge in this figure means that all vertices in one endpoint of this edge are connected to all vertices in the other endpoint.
The big clique $C$ is separated into two parts, and the upper part is connected to all $V_i$ with $1 \le i \le s$.
Two yellow vertices in $E_{1,2}$ represent an edge, and they are connected to two end point of this edge, the red vertex in $V_1$ and the blue vertex in $V_2$.
Vertices below the line in $E_{1,2}$ are dummy vertices, and each of them is connected to all vertices in the lower part of $C$.
}
\label{fig:WHardness_CC}
\end{figure}
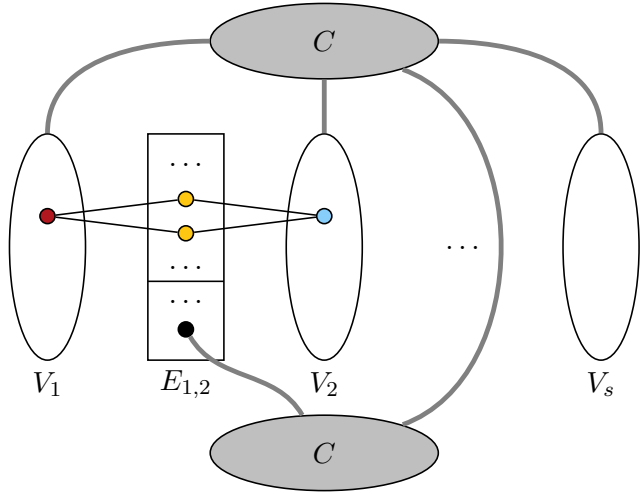

\begin{proof}
We present a parameterized reduction from the \W{1}-hard problem \MCC\ parameterized by the solution size~\cite{cygan2015parameterized}.
In \MCC, we are given an integer $s$ and a $s$-colorable graph with color classes $V_1, V_2, \dots, V_s$, and the task is to find a clique of size $s$ containing one vertex from each color. 
Let $(G=(V,E),s)$ be an instance of \MCC. 
The edge set $E$ can be partitioned into $\binom{s}{2}$ subsets: $E_{i,j}=\{v_iv_j|v_i \in V_i, v_j \in V_j\}, 1\leq i< j\leq s$.
We create an instance $(G'=(V',E'),b,x,k)$ of \CKC\ as follows.
\begin{compactitem}
\item Denote $k=\max_{1\leq i< j\leq s} 2|E_{i,j}|$, $n=|V|$ and set $b=s$, $x=N-N'$ where $N=2n^4+k+s+n+k\binom{s}{2}$ is the number of vertices in $G'$ we will construct and $N'=s+k\binom{s}{2}$.
\item For every $V_i,i=1,2,\dots,s$, create a clique $C_i$ in $G'$, which contains all vertices in $V_i$. 
\item For every $E_{i,j}, 1\leq i< j\leq s$, create a clique $C_{i,j}$ of size $k$ in $G'$, which contains 2 copies of vertices $v_{i,j},{v_{i,j}}'$ for every edge $v_iv_j$ in $E_{i,j}$ and $k-2|E_{i,j}|$ more dummy vertices.
\item For every edge $v_iv_i \in E_{i,j}$, add 4 edges $v_{i,j}v_i$, $v_{i,j}v_j$, ${v_{i,j}}'v_i$ and ${v_{i,j}}'v_j$ in $G'$. 
\item Create a clique $C$ of size $2n^4+k+s$. 
Add edges between vertices in $C_i,i=1,2,\dots,s$ and $k+s$ vertices in $C$. 
For every dummy vertex in $C_{i,j}, 1\leq i< j\leq s$, we add edges between this vertex and two distinct vertices in $C$.
The size of $C$ is large enough such that no pair of edges between $E_{i,j}$ and $C$ share the same end point in $C$.
\end{compactitem}
Notice that the clique cover number of $G'$ is $s+\binom{s}{2}+1$. The construction is illustrated in \autoref{fig:WHardness_CC}.
We claim that there is a multicolored clique of size $s$ in $G$ if and only if~$(G'=(V',E'),b,x,k)$ is a yes-instance.

$\Rightarrow:$ If there is a multicolored clique $C'$ with vertex set $S$ of size $s$ in $G$, 
we show in the following that the $k$-core of $G-S$ has size at most $x$. 
Since $b=s$ and $N'=s+k\binom{s}{2}$, it suffices to show that all edge cliques $C_{i,j}$ collapse. 
For any clique $C_{i,j}$, every vertex in this clique has vertex degree $k+1$, since it has $k-1$ neighbors in $C_{i,j}$ and 2 neighbors in $C_i$ and $C_j$ (or $C$ for dummy vertices).
For any $C_{i,j}$, suppose $v_i,v_j \in S$, where $v_i \in V_i$ and $v_j \in V_j$, since $C'$ is a clique, there is an edge $v_iv_j$ in $G$ , hence both $v_{i,j}$ and ${v_{i,j}}'$ connected to $v_i$ and $v_j$ in $G'$.
After deleting $v_i$ and $v_j$, both $v_{i,j}$ and ${v_{i,j}}'$ will be collapsed, which then make all remaining vertices in $C_{i,j}$ collapse.
Therefore all edge cliques $C_{i,j}$ will collapse after deleting~$S$.

$\Leftarrow:$ Suppose $(G'=(V',E'),b,x,k)$ is a yes-instance, we need to show there is a multicolored clique of size $s$ in $G$.
Let $S$ be the deleted vertex set of size $b$ and let $S'$ be the set of all collapsed vertices. 
Since $N'=s+k\binom{s}{2}$, we have $|S'| \geq k\binom{s}{2}$.
Notice that in the subgraph of $G'$ induced by $C$ and all $C_i$'s all vertices in $C_i$ have vertex degree at least $k+s$ and all vertices in $C$ have degree at least $2n^4+k+s-1$. Therefore, by \autoref{obs:high_deg}, these vertices will never collapse and only vertices in $C_{i,j}, 1\leq i< j\leq s$ can collapse.
Since the number of vertices in $C_{i,j}, 1\leq i< j\leq s$ is exactly $k\binom{s}{2}$, we have all vertices in $C_{i,j}, 1\leq i< j\leq s$ will collapse and $S$ only contains vertices from $C$ and $C_i, 1\leq i \leq s$.

Suppose $S$ contains $t$ vertices from $C$ and $s-t$ vertices from $C_i, 1\leq i \leq s$.
On one hand, for every clique $C_{i,j}$, the first vertex to collapse must connect to two vertices from $S$, so overall there must be $2\binom{s}{2}$ edges between all such vertices and $S$.
On the other hand, each vertex in $S \cap C$ can provide at most one such edge and each vertex in $S$ from $C_i, 1\leq i \leq s$ can provide at most $s-1$ such edges, so overall the number is strictly less than $2\binom{s}{2}$ if $t>0$.
Since all cliques $C_{i,j}, 1\leq i< j\leq s$ will collapse, we have $t=0$ and $S$ only contains vertices from $C_i$.

So the firstly collapsed vertex $v_{i,j}$ in $C_{i,j}$ must connect to two vertices from $S$, one $v_i$ from $C_i$ and another $v_j$ from $C_j$.
This means $S$ contains exactly one vertex $v_i$ from each $C_i$ and each pair of vertices $v_i$ and $v_j$ connect to at least one common vertex $v_{i,j}$ in $C_{i,j}$, which means $v_i$ and $v_j$ are connected in $G$.
Therefore, all vertices from $S$ form a clique of size $s$ in~$G$.
\end{proof}

On the positive side, we sketch a dynamic program on the tree decomposition of the input graph $G$ which implies that \CKC\ is in \FPT\ when parameterized by the treewidth of the input graph.
\begin{prop}
\CKC\ is in \FPT\ when parameterized by the treewidth of the input graph.
\end{prop}
\begin{proof}[Proof Sketch]
Observe that either $k \le \tw(G)$ or the $k$-core is (already) empty and we can answer Yes.
Hence, for the rest of the proof we assume that $k \le \tw(G)$. We assume we are given a nice tree decomposition of $G$~\cite{Kloks94,bodlaender2016c}
and use dynamic programming on the nice tree decomposition of $G$.
The indices of the table are formed for each bag of the decomposition by the number of vertices of the solution already forgotten, the number of vertices in the core already forgotten, a partition of the bag into three set $B$, $X$, and $Q$, an (elimination) order for the vertices in $Q$, and for each vertex in $Q$ the number of its neighbors in $X$ or higher in the order. This number is always in $0,\ldots, k-1$, as otherwise it would not be possible to eliminate the vertex.

The set $B$ represents the partial solution (or rather its intersection with the bag), i.e., the vertices to be deleted.
The set $X$ represents the vertices which (are free to) remain in the core.
The vertices in $Q$ should collapse after removing the vertices of the solution and the collapse of the vertices preceding them in the order.

There are $3^{\tw(G)} \cdot (\tw(G))^{O(\tw(G))} \cdot k^{\tw(G)} = (\tw(G))^{O(\tw(G))}$ possible indices for each bag. Hence the slightly superexponential running time of $(\tw(G))^{O(\tw(G))} \cdot n^{O(1)}$ follows.
\end{proof}

Using monadic second order (MSO) logic formulas, we can show that for a smaller structural parameter, namely the cliquewidth of the input graph, we can also obtain positive results. Here however, we can only show fixed-parameter tractability for the combination of the cliquewidth of the input graph with $k$ and either $b$, $x$, $n-x$, or $n-b$.
\begin{prop}
\label{thm:MSO}
\CKC\ is in \FPT\ when parameterized by the cliquewidth of the input graph combined with $k$ and either $b$, $x$, $n-x$, or $n-b$. 
\end{prop}
\begin{proof}
We first develop, for a fixed $k$ a formula \[ \mathtt{core}(G,B,X),\] which should express that the set $X$ contains the whole $k$-core of the graph $G - B$. The formula thus says that no graph induced by a set larger than $X$, but not containing anything from $B$ is a core. In other words, each such graph contains a vertex of degree at most $k-1$, i.e., not having $k$ distinct neighbors. For that purpose we use the following subformula:
\begin{align*}
\mathtt{smalldeg}_k(v,A) &= \forall x_1 \in V \forall x_2 \in V \ldots \forall x_k \in V \left(\bigwedge_{i=1}^k (x_i \in A \wedge \mathtt{adj}(v, x_i))\right) \\
&\implies \bigvee_{1\le i < j \le k} x_i=x_j 
\end{align*}
Now the sought formula is
\[ \mathtt{core}(G,B,X) = \forall A \subseteq V (A \cap B = \emptyset \wedge X\subseteq A) \implies \exists v \in V (v \in A) \wedge \mathtt{smalldeg}_k(v,A).\]
This formula is of length $O(k)$. Combined with some of the following formulae it gives the result for all parameter combinations promised.
The following formula bounds a set $S$ passed to be of size at most $s$:
\[ \mathtt{sizeatmost}_s(S) = \forall x_1 \in V \forall x_2 \in V \ldots \forall x_s \in V \forall x_{s+1} \in V(\bigwedge_{i=1}^{s+1} x_i \in S ) \implies \bigvee_{1\le i < j \le s+1} x_i=x_j\]
The next formula bounds the size to at most $n-s$:
\[ \mathtt{sizeatmost}_{n-s}(S) = \exists x_1 \in V \exists x_2 \in V \ldots \exists x_s \in V (\bigwedge_{i=1}^s x_i \notin S ) \wedge (\bigwedge_{1\le i < j \le s} x_i \neq x_j) \]
Both these formulae have length $O(s^2)$.

Now the result follows from the theorem of Courcelle et al.~\cite{CourcelleMR00}.
\end{proof}

In the remainder of this section, we show that \CKC\ does not admit a polynomial kernel when parameterized by rather large parameter combinations. We first show an OR-cross composition~\cite{bodlaender2014kernelization,cygan2015parameterized} from \CVC.
\begin{theorem}\label{thm:nopkvc}
For all $k\ge 2$ \CKC\ does not admit a polynomial kernel when parameterized by the combination of $b$ and the vertex cover number
of the input graph unless \NoKernelAssume.
\end{theorem}

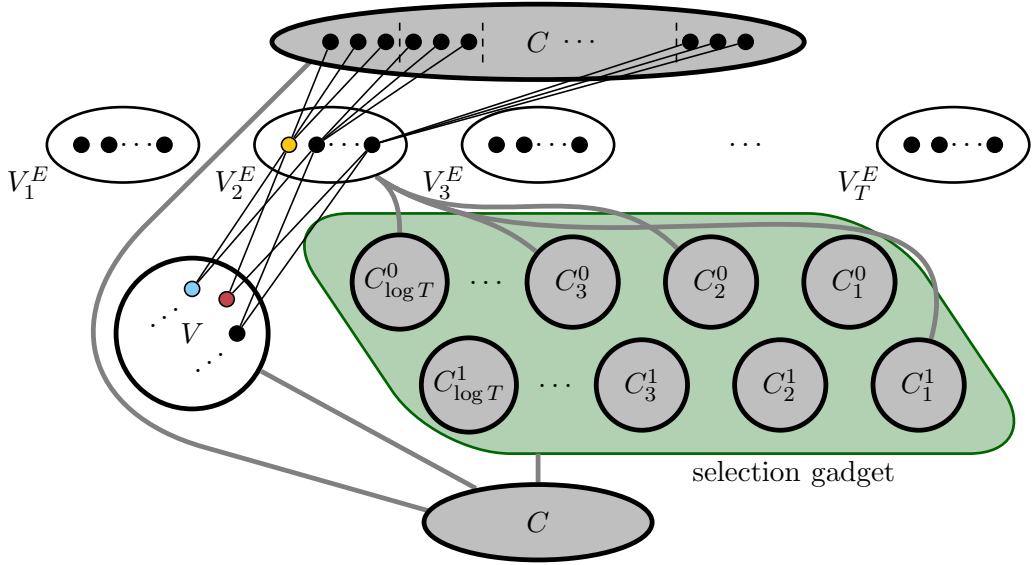
\begin{figure}[t]
\begin{center}
\begin{tikzpicture}[scale=.91, line width=.6pt]
\node[bigclique,minimum width=7cm, minimum height=1cm] (C1) at (0,-.5) {$C$};
\node[smallv] (c11) at (-3,-.5) {};
\node[smallv] (c12) at (-2.6,-.5) {};
\node[smallv] (c13) at (-2.2,-.5) {};
\node[smallv] (c14) at (-1.8,-.5) {};
\node[smallv] (c15) at (-1.4,-.5) {};
\node[smallv] (c16) at (-1,-.5) {};
\node(c18) at (.6,-.5) {$\ldots$};
\node[smallv] (c20) at (2.2,-.5) {};
\node[smallv] (c21) at (2.6,-.5) {};
\node[smallv] (c22) at (3,-.5) {};
\node[eset,label={[label distance=-4pt]190:$V^E_1$}] (E1) at (-6,-2) {};
\node[smallv] (e11) at (-6.6,-2) {};
\node[smallv] (e12) at (-6.2,-2) {};
\node (e13) at (-5.8,-2) {$\ldots$};
\node[smallv] (e14) at (-5.4,-2) {};
\node[eset,label={[label distance=-4pt]190:$V^E_2$}] (E2) at (-3,-2) {};
\node[smallv,fill=lipicsyellow] (e21) at (-3.6,-2) {};
\node[smallv] (e22) at (-3.2,-2) {};
\node (e23) at (-2.8,-2) {$\ldots$};
\node[smallv] (e24) at (-2.4,-2) {};
\node[eset,label={[label distance=-4pt]190:$V^E_3$}] (E3) at (0,-2) {};
\node[smallv] (e31) at (-0.6,-2) {};
\node[smallv] (e32) at (-0.2,-2) {};
\node (e33) at (0.2,-2) {$\ldots$};
\node[smallv] (e34) at (0.6,-2) {};
\node (E4) at (3,-2) {$\ldots$};
\node[eset,label={[label distance=-4pt]190:$V^E_T$}] (ET) at (6,-2) {};
\node[smallv] (e41) at (5.8,-2) {};
\node[smallv] (e42) at (5.4,-2) {};
\node (e43) at (6.2,-2) {$\ldots$};
\node[smallv] (e44) at (6.6,-2) {};
\node[vset,minimum width=2cm, minimum height=2cm, line width=1.8pt] (V) at (-5,-4.75) {$V$};
\node[smallv,fill=ourred!80] (v1) at (-4.5,-4.25) {};
\node[smallv,fill=ourblue] (v2) at (-5,-4.1) {};
\node[smallv] (v3) at (-4.35,-4.75) {};
\node (v4) at (-5.4,-4.4) {\reflectbox{$\ddots$}};
\node (v5) at (-4.75,-5.05) {\reflectbox{$\ddots$}};
\node[sclique] (S00) at (4.5,-4) {$C^0_1$};
\node[sclique] (S01) at (5.5,-5.5) {$C^1_1$};
\node[sclique] (S10) at (2.5,-4) {$C^0_2$};
\node[sclique] (S11) at (3.5,-5.5) {$C^1_2$};
\node[sclique] (S20) at (.5,-4) {$C^0_3$};
\node[sclique] (S21) at (1.5,-5.5) {$C^1_3$};
\node (S30) at (-.75,-4) {$\ldots$};
\node (S31) at (0.25,-5.5) {$\ldots$};
\node[sclique] (ST0) at (-2,-4) {$C^0_{\log T}$};
\node[sclique] (ST1) at (-1,-5.5) {$C^1_{\log T}$};
\node[bigclique,minimum width=3cm, minimum height=1cm] (C2) at (0,-7.5) {$C$};

\path (e21) edge (c11);
\path (e21) edge (c12);
\path (e21) edge (c13);
\path (e22) edge (c14);
\path (e22) edge (c15);
\path (e22) edge (c16);
\path (e24) edge (c20);
\path (e24) edge (c21);
\path (e24) edge (c22);

\path (e21) edge (v1);
\path (e21) edge (v2);
\path (e22) edge (v2);
\path (e22) edge (v3);
\path (e24) edge (v1);
\path (e24) edge (v3);

\path[dashed] (-2,-.8) edge (-2,-.2);
\path[dashed] (-.8,-.8) edge (-.8,-.2);
\path[dashed] (2,-.8) edge (2,-.2);

\begin{pgfonlayer}{bg}
\path (C2) edge[thickgray] (V);
\draw[rounded corners=10mm, fill=ourgreen!30] (-3.83,-3) -- (5,-3) -- (7.33,-6.5) -- (-1.5,-6.5) -- cycle;
\draw[rounded corners=10mm, line width=1pt,draw=ourgreen] (-3.83,-3) -- (5,-3) -- (7.33,-6.5) -- (-1.5,-6.5) -- cycle;
\node (t) at (3.7,-6.8) {selection gadget};
\path (C2) edge[thickgray] (0,-6.5);
\draw[thickgray, rounded corners=30mm] (-1,-7.5) -- (-8,-5.5) -- (c11);


\path (E2) edge[in=75, out=-35,thickgray] (S01);
\path (E2) edge[in=135, out=-35,thickgray] (S10);
\path (E2) edge[in=140, out=-35,thickgray] (S20);
\path (E2) edge[in=90, out=-35,thickgray] (ST0);

%
\end{pgfonlayer}

\end{tikzpicture}
\end{center}
\caption{Illustration of the OR-cross composition from \CVC\ to \CKC\ with $k=5$.
The selection gadget is all the circles contained in the green box.
Every gray edge in this figure means that all vertices in one endpoint of this edge are connected to all vertices in the other endpoint.
Every vertex in $V_i^E$ connects to two endpoints of its corresponding edge in~$G_i$.
For example, the yellow vertex in $V_2^E$ is connected to the blue and the red vertex in $V$, which represents the two endpoints of the corresponding edge in $G_2$.
The big clique $C$ is separated into two parts. 
Every vertex in $V_i^E$ is connected to $k-2$ vertices in the upper part of $C$.
Since $k=5$, the yellow vertex in $V_2^E$ is connected to three vertices in $C$.
Vertices contained in thick outlined vertex sets form a vertex cover.
To keep the picture simple, edges that contain vertices from $V_i^E$ with $i\neq 2$ are not depicted.
}
\label{fig:nopkvc}
\end{figure}

\begin{proof}
We apply an OR-cross composition~\cite{bodlaender2014kernelization,cygan2015parameterized} from the \NP-hard problem \CVC~\cite{GJ79}.
In \CVC, we are given a 3-regular graph $G$ an integer $s$, and the task is to find a vertex subset of size at most $s$ which contains at least one endpoint of each edge of $G$.

We say an instance of \CVC is \emph{malformed} if the string does not represent a pair $(G,s)$, where $G$ is a 3-regular graph and $s$ is a non-negative integer. It is \emph{trivial}, if $s \ge |V(G)|$.
We define the equivalence relation $\mathcal{R}$ as follows: all malformed instances are equivalent, all trivial instances are equivalent and two well-formed non-trivial instances $(G,s)$ and $(G',s')$ are $\mathcal{R}$-equivalent if $|V(G)|=|V(G')|$ and $s=s'$. 
Observe that $\mathcal{R}$ is a polynomial equivalence relation.

Let the input consist of $T$ $\mathcal{R}$-equivalent instances of \CVC.
If the instances are malformed or trivial, we return a constant size no- or yes- instance of \CKC, respectively.
Let $(G_i,s)_{0 \leq i \leq T-1}$ be well-formed non-trivial $\mathcal{R}$-equivalent instances of \CVC.
Since all instances have the same size of the vertex set, we can assume they share the same vertex set~$V=\{v_1,v_2,\dots,v_n\}$ .
We assume $T$ to be a power of 2, as otherwise we can duplicate some instances.
Now we create an instance $(G,b,x,k)$ of \CKC\ for some arbitrary but fixed $k \geq 2$ as follows.
\begin{compactitem}
\item Set $b=s+2ks\log{T}$ and $x=N-N'$, where $N=n+\frac{3}{2}nT+4ks\log{T}+k+s+\frac{3}{2}n(k-2)$ is the number of all vertices in graph $G$ we will construct and $N'=s+2ks\log{T}+\frac{3}{2}n$. 
\item First for every vertex $v_i$ in $V$, create a vertex $v_i$ in $G$.
\item For every edge set $E(G_i)$, create a vertex set $V_i^E$ in $G$, in which every vertex $v_{p,q}$ represents an edge $v_pv_q$ in $E(G_i)$.
Then we have $T$ of these vertex sets and each set has $\frac{3}{2}n$ vertices. 
\item For every edge $v_pv_q$ in $E(G_i)$, add 2 edges $v_{p,q}v_p$ and $v_{p,q}v_q$ in $G$.
\item Now create the selection gadget in $G$. 
It contains $\log{T}$ pairs of cliques $C_i^d$ $(1 \leq i \leq \log{T}, d \in \{0,1\})$, and all of them have the same size of $2ks$.
For every vertex set $V_i^E$, let $i = {(d_{\log T-1}d_{\log T-2} \dots d_0)}_2$ be the binary representation of the index $i$, where $d_j \in \{0, 1\}$ for $0 \leq j \leq \log T-1$ and we add leading zeros so that the length of the representation is exactly $\log{T}$, we add edges between all vertices in $V_i^E$ and all vertices in $C_j^{d_j}$ $(0 \leq j \leq \log T-1)$.
\item Finally we create a clique $C$ with $|C|=k+b+\frac{3}{2}n(k-2)$, which contains two parts of vertices. 
The first part contains $k+b$ vertices and each of them connects to all vertices in~$V$ and all $C_j^d$ with $0 \leq j \leq \log T-1$ and $d \in \{0,1\}$.
In other words, all vertices in $V$ and $C_j^d$ connect to these $k+b$ vertices in $C$.
The second part of $\frac{3}{2}n(k-2)$ vertices are connected to vertices in $V_i^E$ in the following way.
For every vertex $v_{p,q}$ in $V_i^E$, add edges between $v_{p,q}$ and $k-2$ vertices in $C$.
We can make sure that all vertices in the same $V_i^E$ connect to different vertices in $C$.
In other words, every vertex in the second part of $C$ connects to exactly one vertex in every~$V_i^E$.
\end{compactitem}
Notice that the vertex cover number of $G$ is $n+4ks\log{T}+k+b+\frac{3}{2}n(k-2)$. The construction is illustrated in \autoref{fig:nopkvc}.
We now show that at least one instance $(G_i,s)$ is a yes-instance if and only if the instance $(G,b,x,k)$ of \CKC\ constructed above is a yes-instance.

$\Rightarrow:$ If instance $(G_i,s)$ is a yes-instance, which means there is a vertex subset $V^*$ of size~$s$ that covers all edges in $G_i$, then we delete the corresponding $s$ vertices in $G$ and all vertices in $C_j^{d_j}(0 \leq j \leq \log T-1)$, where $i = {(d_{\log T-1} \dots d_0)}_2$ is the binary representation of $i$.
So far, we deleted $s+2ks\log{T}$ vertices, and all vertices in $V_i^E$ will collapse, since they just have at most $k-1$ edges remaining, $k-2$ of which connect to vertices in $C$ and at most one to vertices in $V$.
Therefore, the number of remaining vertices is $x$ and instance $(G,b,x,k)$ is a yes-instance.

$\Leftarrow:$ If instance $(G,b,x,k)$ is a yes-instance, we need to show there is at least one instance which has a vertex cover of size $s$. 
Let $S$ be the deleted vertex subset of size $b$ and let $S'$ be the set of all collapsed vertices.
Since $N'=s+2ks\log{T}+\frac{3}{2}n$, we have $|S'| \geq \frac{3}{2}n$.
In the subgraph of $G$ induced by $V$, $C$ and all $C_i^d$'s all vertices in $V$, $C_i^d$ and $C$ have degree larger than $k+b$. Hence, by \autoref{obs:high_deg} they will not collapse,
all collapsed vertices come from $V_i^E(0 \leq i \leq T-1)$.

Then we show that all collapsed vertices can only come from one single~$V_i^E$ for some~$i$. 
Suppose two vertices $v$ and $v'$ from different sets of $V_i^E$ $(0 \leq i \leq T-1)$ collapse after deleting~$S$, then there is at least one pair of cliques $C_{j_0}^0$ and $C_{j_0}^1$ such that $v$ is connected to all vertices in~$C_{j_0}^0$ and~$v'$ is connected to all vertices in $C_{j_0}^1$.
To make $v_1$ collapse, at least~$2ks\log{T}-(k-1)$ vertices from the corresponding cliques in the selection gadget need to be deleted.
Then to make $v_2$ collapse, at least $2ks-(k-1)$ vertices from $C_i^1$ need to be deleted.
Therefore, at least $2ks\log{T}+2ks-2(k-1)$ vertices need to be deleted, which is strictly larger than $b$.
This means the collapsed vertices come from one single $V_i^E$.
Since $|S'| \geq \frac{3}{2}n$ and $|V_i^E|=\frac{3}{2}n$, we have $S'=V_i^E$ and $S \cap V_i^E = \emptyset$ for some $i$.

Then we consider the vertex set $S$. 
We know that after deleting $S$, all vertices in $V_i^E$ collapse.
Denote $V_I$ the vertex set of all vertices in $C_j^{d_j}$ $(0 \leq j \leq \log T-1)$, where $d_j$ is the binary numbers in the binary representation $i = {(d_{\log T-1} \dots d_0)}_2$.
Since every vertex in $V_I$ is connected to all vertices in $V_i^E$, hence to make $V_i^E$ collapse, it is always better to choose vertices from $V_I$ than any other vertex.
If $S$ does not contain all vertices from $V_I$, we can update $S$ by replacing any $|V_I \setminus S|$ vertices in $S$ with vertices in $V_I \setminus S$.
Then $V_I \subseteq S$.

Suppose there is a vertex $v_{p,q}$ in $V_i^E$ such that both $v_p$ and $v_q$ are not in $S \cap V$, then $S$ contains at least one vertex $v_c$ in $C$ connected to $v_{p,q}$, as otherwise $v_{p,q}$ has degree at least~$k$ and will not collapse.
We update $S$ by replacing $v_c$ with $v_p$.
This will not influence the size of $S$ and more importantly, this will not influence the collapsed set $S'=V_i^E$, since $v_c$ in~$C$ is connected to only one vertex $v_{p,q}$ in $V_i^E$, and $v_{p,q}$ will still collapse under the new $S$.
By updating $S$ in the same way for other vertices in $V_i^E$ not covered by vertices in $S \cap V$, we get a vertex set $S \cap V$ which covers all vertices in $V_i^E$ at least once.
And $|S \cap V| \leq s$, since~$V_I \subseteq S$ and$|V_I|=2ks\log{T}$.
This corresponds to a vertex cover of size $s$ in $G_i$.
\end{proof}
Lastly, we give a OR-cross composition~\cite{bodlaender2014kernelization,cygan2015parameterized} from \CKC\ onto itself. Note that the parameter combination of the following result is incomparable to \autoref{thm:nopkvc}.
\begin{prop}
\label{thm:npk_wrt_bandwidth}
\CKC\ does not admit a polynomial kernel when parameterized by the combination of $b$, $k$, and the bandwidth
of the input graph unless \NoKernelAssume.
\end{prop}
\begin{proof}
We apply an OR-cross composition~\cite{bodlaender2014kernelization,cygan2015parameterized} from \CKC\ to \CKC.

We say an instance of \CVC is \emph{malformed} if the string does not represent a quadruple $(G,b,x,k)$, where $G$ is a graph and $b$, $x$, and $k$ are non-negative integers. It is \emph{trivial}, if $b \ge |V(G)|$, $x \ge |V(G)|$, or $k \ge |V(G)|$.
We define the equivalence relation $\mathcal{R}$ as follows: all malformed instances are equivalent, all trivial instances are equivalent and two well-formed non-trivial instances $(G_1,b_1,x_1,k_1)$ and $(G_2,b_2,x_2,k_2)$ are $\mathcal{R}$-equivalent if $|V(G_1)|=|V(G_2)|$, $b_1=b_2$, $x_1=x_2$ and $k_1=k_2$. 
Observe that $\mathcal{R}$ is a polynomial equivalence relation.

Let the input consist of $T$ $\mathcal{R}$-equivalent instances of \CKC.
If the instances are malformed or trivial, we return a constant size no- or yes- instance of \CKC, respectively.
Let $(G_i,b_i,x_i,k_i)_{1 \leq i \leq T}$ be well-formed non-trivial $\mathcal{R}$-equivalent instances of \CKC.
Since all instances have the same $b_i,x_i,k_i$ and all $G_i$ with $1 \leq i \leq T$ have the same size of vertex set, we denote $b=b_i,x=x_i,k=k_i$ and $n=|V(G_i)|$.
If $n \le 3$, then we solve all the instances in $O(T)$ time and output a constant-size instance with the appropriate answer.

Now we create an instance $(G,b',x',k')$ of \CKC.
We start by making $G$ a disjoint union of all $G_i$.
For each $i \in \{1, \ldots, T\}$, we add to $G$ two clique $C_i$ and $C_i'$, each of size $n^2$, and add all edges between $G_i$ and $C_i$ and between $C_i$ and $C_i'$.
Note that $G$ contains $T(n+2n^2)$ vertices in total. 
Set $b'=b+n^2$, $k'=k$ and $x'=(n+2n^2)(T-1)+n^2+x$. 

Notice that the bandwidth of $G$ is upper bounded by $2n^2$.
We now show that at least one instance $(G_i,b,x,k)$ is a yes-instance if and only if the instance $(G,b',x',k')$ of \CKC\ constructed above is a yes-instance.

$\Rightarrow:$ If instance $(G_i,b,x,k)$ is a yes-instance, then there is a subset $S \subseteq V(G_i)$ such that the $k$-core of $G_i-S$ has size at most $x$.
If we delete all vertices from $S$ and the whole clique $C_i$ in $G$, then at most $x$ vertices in $G_i$ remain.
Therefore the $k$-core of $G-S-V(C_i)$ has size at most $x'$.

$\Leftarrow:$ If instance $(G,b',x',k')$ is a yes-instance, then let $S$ be the set of deleted vertices of size at most $b'$ and let $S'$ be the set of all collapsed vertices.
Since the degree of vertices in $C_i$ and $C_i'$ is at least $2n^2-1$ which is more than $2n \ge b+k$ for $n \ge 3$, these vertices will never collapse by \autoref{obs:high_deg}.
So $S'$ only contains vertices from $\bigcup_{i=1}^T V(G_i)$.
Furthermore, it is impossible for two vertices $v_i$ and $v_j$ from different sets $V(G_i)$ and $V(G_j)$ to collapse after deleting~$S$.
Indeed, suppose they do collapse, then $|S \cap C_i| \geq n^2-k$ and $|S \cap C_j| \geq n^2-k$, which means $|S| \geq 2n^2-2k \ge 2n^2-2n > n^2+n \ge n^2+b=b'$, where the middle inequality follows from $n \ge 3$.
Therefore $S'$ only contains vertices from a single graph, say $G_i$.

Since $G$ contains $T(n+2n^2)$ vertices in total, $x'=(n+2n^2)(T-1)+n^2+x$, and $b'=b+n^2$, we have $|S'| \geq n-b-x$.
To make vertices in $G_i$ collapse, it is always better to choose vertices from $C_i$ into $S$, as vertices from $C_i$ connect to all vertices in $G_i$.
Thus we can assume $C_i \subseteq S$.
Then $S \cap V(G_i) \le b$.
If $|V(G_i) \setminus (S \cup S')| > x$, then we can remove vertices from $S \setminus (C_i \cup G_i)$ and add vertices from $G_i \setminus (S \cup S')$ to $S$ till $S \cap V(G_i) = b$.
This will not influence the collapsed vertices in $S'$.
Then we get a vertex set $S_i=S \cap V(G_i)$, and the $k$-core of $G_i-S_i$ has size at most $x$.
\end{proof}

\section{Conclusion}
Our results highlight a dichotomy in the computational complexity of \CKC\ for $k\le 2$ and $k\ge 3$. Along the way, we correct some inaccuracies in the literature concerning the parameterized complexity of \CKC\ with $k=3$ and $x=0$ and give a simple single exponential linear time parameterized algorithm for \textsc{Feedback Vertex Set}. We further investigate the parameterized complexity with respect to several structural parameters of the input graph. As a highlight we show that \CKC\ does not admit polynomial kernels for rather large parameter combinations. We leave the complexity of \CKC\ when parameterized by the cliquewidth of the input graph open.

\paragraph*{Acknowledgments.}
This research was initiated at the annual research retreat of the Algorithmics and Computational Complexity (AKT) group of TU~Berlin, held in Darlingerode, Germany, from March~19 till March~23, 2018. The authors would like to thank Anne-Sophie Himmel for initial discussions leading to the results in this paper.
\bibliography{bib}


\end{document}